\documentclass[letterpaper,11pt]{article}
\usepackage[letterpaper,margin=1in]{geometry}

\usepackage{microtype}
\usepackage{xcolor}
\usepackage{xspace}
\usepackage{amsmath}
\usepackage{amssymb}
\usepackage{amsthm}
\usepackage{enumitem}
\usepackage{textgreek}
\usepackage{graphicx}
\usepackage{framed}
\usepackage[small]{caption}
\usepackage{booktabs}
\usepackage{tikz-cd}
\usepackage[unicode]{hyperref}
\usepackage{doi}
\usepackage[capitalize, nameinlink]{cleveref}
\usepackage[linesnumbered,ruled,vlined,algo2e]{algorithm2e}
\Crefname{remark}{Remark}{Remarks}
\Crefname{observation}{Observation}{Observations}

\theoremstyle{plain}
\newtheorem{theorem}{Theorem}
\newtheorem{lemma}[theorem]{Lemma}

\newtheorem{corollary}[theorem]{Corollary}

\theoremstyle{definition}
\newtheorem{definition}[theorem]{Definition}

\newtheorem{observation}[theorem]{Observation}

\theoremstyle{remark}

\DeclareMathOperator{\poly}{poly}

\newcommand{\LOCAL}{\ensuremath{\mathsf{LOCAL}}\xspace}
\newcommand{\CONGEST}{\ensuremath{\mathsf{CONGEST}}\xspace}

\definecolor{darkgreen}{rgb}{0,0.5,0}
\definecolor{darkred}{rgb}{0.4,0,0}
\hypersetup{
	colorlinks=true,
	linkcolor=darkred,
	citecolor=darkgreen,
	filecolor=black,
	urlcolor=[rgb]{0,0.1,0.5},
	pdftitle={Node and Edge Averaged  Complexities of Local Graph Problems},
	pdfauthor={Alkida Balliu, Mohsen Ghaffari, Fabian Kuhn, Dennis Olivetti}
}

\newenvironment{myabstract}
{\list{}{\listparindent 1.5em%
		\itemindent    \listparindent
		\leftmargin    1cm
		\rightmargin   1cm
		\parsep        0pt}%
	\item\relax}
{\endlist}

\newenvironment{mycover}
{\list{}{\listparindent 0pt
		\itemindent    \listparindent
		\leftmargin    1cm
		\rightmargin   1cm
		\parsep        0pt}%
	\raggedright
	\item\relax}
{\endlist}

\newcommand{\myemail}[1]{\,$\cdot$\, {\small #1}}
\newcommand{\myaff}[1]{\,$\cdot$\, {\small #1}\par\smallskip}

\newcommand{\set}[1]{\left\{#1\right\}}

\newcommand{\calA}{\mathcal{A}}

\newcommand{\calG}{\mathcal{G}}

\newcommand{\calP}{\mathcal{P}}

\DeclareMathOperator{\E}{\mathbb{E}}

\newcommand{\nodeavg}{\mathsf{AVG}_V}
\newcommand{\edgeavg}{\mathsf{AVG}_E}
\newcommand{\wnodeavg}{\mathsf{AVG}_V^{\boldsymbol{w}}}

\newcommand{\nodeexp}{\mathsf{EXP}_V}

\newcommand{\nodeworst}{\mathsf{WORST}_V}

\graphicspath{{figs/}}

\begin{document}

 \setcounter{page}{0}
 \thispagestyle{empty}

\begin{mycover}
     {\huge\bfseries Node and Edge Averaged  Complexities \\of Local Graph Problems \par}
     \bigskip
     \bigskip
     \bigskip

     \textbf{Alkida Balliu}
     \myemail{alkida.balliu@gssi.it}
     \myaff{Gran Sasso Science Institute}

     \textbf{Mohsen Ghaffari}
     \myemail{ghaffari@inf.ethz.ch}
     \myaff{ETH Zurich}

    \textbf{Fabian Kuhn}
     \myemail{kuhn@cs.uni-freiburg.de}
     \myaff{University of Freiburg}

     \textbf{Dennis Olivetti}
     \myemail{dennis.olivetti@gssi.it}
     \myaff{Gran Sasso Science Institute}
 \end{mycover}
 \bigskip

\begin{myabstract}
  We continue the recently started line of work on the distributed node-averaged complexity of distributed graph algorithms. The node-averaged complexity of a distributed algorithm running on a graph $G=(V,E)$ is the average over the times at which the nodes $V$ of $G$ finish their computation and commit to their outputs. We study the node-averaged complexity for some of the central distributed symmetry breaking problems and provide the following results (among others):
  \begin{itemize}
  \item As our main result, we show that the randomized node-averaged complexity of computing a maximal independent set (MIS) in $n$-node graphs of maximum degree $\Delta$ is at least $\Omega\big(\min\big\{\frac{\log\Delta}{\log\log\Delta},\sqrt{\frac{\log n}{\log\log n}}\big\}\big)$. This bound is obtained by a novel adaptation of the well-known lower bound of Kuhn, Moscibroda, and Wattenhofer [JACM'16]. As a side result, we obtain that the worst-case randomized round complexity for computing an MIS in trees is also $\Omega\big(\min\big\{\frac{\log\Delta}{\log\log\Delta},\sqrt{\frac{\log n}{\log\log n}}\big\}\big)$---this essentially answers open problem 11.15 in the book of Barenboim and Elkin and resolves the complexity of MIS on trees up to an $O(\sqrt{\log\log n})$ factor. We also show that, perhaps surprisingly, a minimal relaxation of MIS, which is the same as $(2,1)$-ruling set, to the $(2,2)$-ruling set problem drops the randomized node-averaged complexity to $O(1)$. 
  
  \item For maximal matching, we show that while the randomized node-averaged complexity is $\Omega\big(\min\big\{\frac{\log\Delta}{\log\log\Delta},\sqrt{\frac{\log n}{\log\log n}}\big\}\big)$, the randomized edge-averaged complexity is $O(1)$. Further, we show that the deterministic edge-averaged complexity of maximal matching is $O(\log^2\Delta + \log^* n)$ and the deterministic node-averaged complexity of maximal matching is $O(\log^3\Delta + \log^* n)$.
  
  \item Finally, we consider the problem of computing a sinkless orientation of a graph. The deterministic worst-case complexity of the problem is known to be $\Theta(\log n)$, even on bounded-degree graphs. We show that the problem can be solved deterministically with node-averaged complexity $O(\log^* n)$, while keeping the worst-case complexity in $O(\log n)$.
  \end{itemize}
\end{myabstract}

\clearpage
\setcounter{page}{0}
\thispagestyle{empty}
\tableofcontents

\clearpage

\section{Introduction}
\label{sec:intro}

The main focus throughout the past four decades of studying distributed algorithms for graph problems has traditionally been on the worst-case round complexity. That is, the round complexity of the algorithm is defined to be the number of rounds until all nodes in the network terminate. This has proved fruitful in many contexts. For example, it allows to cleanly talk about the \emph{locality} of a graph problem: e.g., any deterministic $k$-round algorithm for a given problem shows that any node's output can be determined by a function of the topology induced by the $k$-hop neighborhood of the node~\cite{linial1987LOCAL,peleg00}. More recently~\cite{feuilloley2020long,barenboim2019distributed,chatterjee2020sleeping}, starting with an initial exploration of Feuilloley~\cite{feuilloley2020long}, there has been interest in going beyond this worst-cast measure. As a primary example, this involves asking ``what is the average termination time among the nodes?" There are different arguments for why understanding such node-averaged complexities is valuable in different contexts. It is indicative of the run-time of a typical node~\cite{feuilloley2020long} and it also implies sharper bounds on the overall energy spent in a network~\cite{chatterjee2020sleeping}. In this paper, we investigate the node and edge averaged complexity of some of the most prominent problems in the literature on distributed graph algorithms. 

\subsection{Our Contributions}
\paragraph{Maximal Independent Set.} The maximal independent set (MIS) problem is one of the central problems in the area of distributed graph algorithms. Feuilloley~\cite{feuilloley2020long} showed that Linial's $\Omega(\log^* n)$ lower bound for computing an MIS on $n$-node cycles even applies to node-averaged round complexity, as long as we stick to deterministic algorithms. In contrast, even though Linial's $\Omega(\log^* n)$ worst-case round complexity~\cite{linial1987LOCAL} even holds for randomized algorithms~\cite{Naor91}, when switching to node-averaged complexity, randomized algorithms can easily break this barrier. Indeed for any constant degree graph, one can obtain a \emph{randomized} MIS algorithm with $O(1)$ node-averaged round complexity in a straightforward way: among many others, Luby's algorithm~\cite{alon86,luby86} will remove each node with a constant probability in a single phase of the algorithm and it thus has a node-averaged complexity of $O(1)$ on constant degree graphs. 

One of the main open questions in the study of distributed node-averaged complexity is whether this $O(1)$ bound is achievable in general graphs---this was for instance recently mentioned explicitly by Chatterjee et al.\ in \cite{chatterjee2020sleeping}. It is however worth noting that this particular question is in some sense much older: Luby's analaysis~\cite{luby86} shows that his MIS algorithm removes a constant fraction of the edges per iteration, and it has been open since then whether a distributed MIS algorithm can also remove a constant fraction of the nodes in $O(1)$ rounds.

As one of our main contributions, we refute the possibility of a randomized MIS algorithm with node-averaged complexity $O(1)$. In particular, we give a modification of the KMW lower bound by Kuhn, Moscibroda, and Wattenhofer~\cite{lowerbound,kuhn16_jacm} to show that their bound also applies to the node-averaged complexity of computing an MIS. That is, there is a family of graphs for which the node-averaged complexity of any randomized distributed MIS algorithm is lower bounded by $\Omega\Big(\min\Big\{\frac{\log \Delta}{\log\log \Delta}, \sqrt{\frac{\log n}{\log\log n}}\Big\}\Big)$. We also comment that a randomized MIS algorithm from prior work by Bar-Yehuda, Censor-Hillel, Ghaffari, and Schwartzman~\cite{bar2017distributed}
has a node-averaged complexity of $O\big(\frac{\log \Delta}{\log\log \Delta}\big)$. Hence, at least for graphs of maximum degree $\Delta = 2^{O(\sqrt{\log n\cdot\log\log n})}$, the node-averaged complexity of MIS is now completely resolved. For larger $\Delta$, the problem remains open. This is however also true for the worst-case round complexity and our work closes the gap essentially to where it currently also is for the worst-case complexity of MIS.

As a side result, 
we also obtain the same lower bound for the worst-case MIS complexity in trees, i.e., any randomized MIS algorithm in trees requires time at least $\Omega\Big(\min\Big\{\frac{\log \Delta}{\log\log \Delta}, \sqrt{\frac{\log n}{\log\log n}}\Big\}\Big)$. For $\Delta=\log^{\omega(1)}n$, this improves 
on a recent randomized $\Omega\big(\min\big\{\Delta, \frac{\log\log n}{\log\log\log n}\big\}\big)$ lower bound for the same problem~\cite{balliu2021hideandseek}. The lower bound also almost matches the best known randomized MIS algorithm in trees, which has a worst-case complexity of $O\big(\min\big\{\sqrt{\log n}, \log\Delta +\frac{\log\log n}{\log\log\log n}\big\}\big)$~\cite{ghaffari2016MIS}, and we thus nearly resolve Open Problem 11.15 in the book of Barenboim and Elkin~\cite{barenboim15}. 

\paragraph{Maximal Matching.} A basic problem that is closely related to MIS is the problem of computing a maximal matching. We show that also for maximal matching, the randomized node-averaged complexity has an $\Omega\Big(\min\Big\{\frac{\log \Delta}{\log\log \Delta}, \sqrt{\frac{\log n}{\log\log n}}\Big\}\Big)$ lower bound. In this case, the bound follows almost immediately from the KMW lower bound construction that has been used for the approximate maximum matching problem in \cite{kuhn16_jacm}. Note, however, that a more apt comparison with MIS would be to consider the edge-averaged complexity of maximal matching. This is because when computing a matching, the output is on the edges rather than on the nodes. 

Recall that for a graph $G=(V, E)$, its line graph $H=(E, E')$ is the graph where we put one vertex for each edge of $G$ and we connect two of those vertices of $H$ if their corresponding edges in $G$ share an endpoint. Any maximal matching of a graph $G$ is simply an MIS of the line graph of $G$. Consequently, the node-averaged complexity of this MIS problem is equal to the edge-averaged complexity of the maximal matching problem. We show that, unlike in the general case, the MIS problem on line graphs has an $O(1)$ node-averaged complexity. Concretely, a close variant of Luby's randomized classic algorithm~\cite{alon86,luby86,IsraelI86} provides a maximal matching algorithm with edge-averaged complexity $O(1)$. 

We also provide results for the deterministic averaged complexity of maximal matching by giving an algorithm that achieves an $O(\log^2 \Delta + \log^* n)$ edge-averaged complexity and an $O(\log^3 \Delta + \log^* n)$ node-averaged complexity. The algorithm is obtained by adapting deterministic matching algorithms developed in \cite{fischer2020improved,ahmadi2018distributed}. We note that the current best deterministic worst-case complexity of maximal matching is $O(\log^2\Delta\cdot \log n)$~\cite{fischer2020improved} and any improvement on the edge-averaged complexity would thus most likely also improve the state of the art of the worst-case round complexity.

\paragraph{Ruling Set.} Faced with the $\Omega\big(\min\big\{\frac{\log \Delta}{\log\log \Delta}, \log_\Delta n\big\}\big)$ lower bound on the node-averaged complexity of the MIS problem, it is natural to wonder if any relaxation of the problem admits a better node-averaged complexity. A natural relaxation of MIS that has been studied quite intensively is ruling sets. For positive integers $\alpha$ and $\beta$, an $(\alpha,\beta)$-ruling set is a set $S$ of nodes such that any two nodes in $S$ are at distance at least $\alpha$, and for every node not in $S$, there is a node in $S$ at distance at most $\beta$~\cite{awerbuch89}. An MIS is therefore a $(2,1)$-ruling set. We show that, perhaps surprisingly, even relaxing the MIS problem only slightly to the problem of computing a $(2,2)$-ruling set completely avoids the lower bound. The $(2,2)$-ruling set problem (i.e., the problem of computing an independent set $S$ such that any node not in $S$ has a node in $S$ within distance $2$) admits a randomized algorithm with a node-averaged complexity of $O(1)$. It is plausible that in many applications of maximal independent sets (e.g., if an MIS algorithm is used as a subroutine in a higher-level algorithm), one could also work with the weaker $(2,2)$-ruling sets. Doing this might lead to an algorithm that is considerably faster from a node-averaged perspective.

We also study the node-averaged complexity of deterministic ruling set algorithms. We give algorithms with $O(\log^* n)$ node-averaged complexity to compute $(2, O(\log \Delta))$-ruling sets and $(2, O(\log \log n))$-ruling sets. Contrast this with the worst-case deterministic round complexity measure: The best known deterministic algorithm for computing a $(2, \beta)$-ruling set has a round complexity of  $O(\beta \Delta^{2/(\beta+1)} + \log^* n)$~\cite{schneider2013symmetry} and it is known that any deterministic $(2, \beta)$-ruling set algorithm requires at least $\Omega\big(\min\big\{\beta \Delta^{1/\beta},\log_{\Delta} n \big\}\big)$ rounds~\cite{balliu2021hideandseek}. 

\paragraph{Sinkless Orientation.} Finally, we investigate the node-averaged complexity of the sinkless orientation problem. While the worst-case time of deterministic algorithms for computing a sinkless orientation is $\Theta(\log n)$~\cite{brandt2016lower,chang16,ghaffari2017orinetation}, we show that there is a deterministic distributed sinkless orientation algorithm with node-averaged round complexity $O(\log^* n)$.

\subsection{Other Related Work}

As discussed, the explicit study of node-averaged complexity was initiated by Feuilloley in \cite{feuilloley2020long}. He in particular proved that the deterministic distributed node-averaged complexity of locally checkable labeling (LCL) problems on $n$-node cycles is asymptotically equal to the worst-case deterministic complexity of the same problems. This implies that the deterministic $\Omega(\log^* n)$-round lower bound for coloring cycles with $O(1)$ colors and for MIS and maximal matching on cycles extends to the node-averaged and to the edge-averaged complexities of those problems. Subsequently, Barenboim and Tzur~\cite{barenboim2019distributed} studied the node-averaged complexity of different variants of the distributed vertex coloring problem. They in particular analyzed the problem as a function of the arboricity of the graph and gave various trade-offs between the achievable number of colors and node-averaged complexity. They are also the first to explicitly observe that the randomized node-averaged complexity of the $(\Delta+1)$-vertex coloring problem is $O(1)$.

One of the practical motivations to look at node-averaged complexity is to optimize the overall energy usage of a distributed system. If we assume that the energy used by a node in a distributed algorithm is proportional to the number of rounds in which the node participates, the node-averaged complexity can be used as a measure for the total energy spent by all nodes (normalized by the total number of nodes). In this context, Chatterjee, Gmyr, and Panduarangan~\cite{chatterjee2020sleeping} introduced the notion of node-averaged awake complexity. In their setting, nodes are allowed to only participate in a subset of the rounds of an algorithm's execution and to sleep in the remaining rounds. The awake complexity of the node is then measured by the number of rounds in which the node is participating in the protocol (i.e., sending and/or receiving messages) and the node-averaged awake complexity denotes the average number of rounds in which nodes are awake during an algorithm. In \cite{chatterjee2020sleeping}, it is shown that there is a randomized a distributed MIS algorithm with node-averaged awake complexity $O(1)$. In \cite{chatterjee2020sleeping}, it is left as an open question whether it is also possible to obtain a distributed MIS algorithm with (regular) node-averaged complexity $O(1)$. As discussed above, we prove that this is not the case. The study of distributed node awake complexity for local graph problems was continued in a recent paper by Barenboim and Maimon~\cite{BarenboimM21}. This paper however studies the worst-case node awake complexity. They show that every decidable problem can be solved by distributed algorithm with node awake complexity $O(\log n)$ and that for a natural family of problems, one can obtain a node  awake complexity of $O(\log\Delta + \log^* n)$. Notions that are closely related to the notion of node awake complexity of \cite{chatterjee2020sleeping,BarenboimM21} have also been studied in the context of radio network (mostly from a worst-case complexity point of view), see e.g., \cite{NakanoO00a,jurdzinski2002efficient,jurdzinski2002energy,KardasKP13,BenderKPY18,chang2019exponential,chang2020energy}.

While there is not a lot of work that explicitly studies the node or edge-averaged complexity of distributed algorithms, many existing distributed algorithms implicitly provide averaged complexity bounds that are stronger than the respective worst-case complexity bounds. This for example leads to the following node-averaged complexities of the $(\Delta+1)$-coloring problem. The randomized algorithms of \cite{luby1993removing,johansson99} are based on the following idea. The nodes pick random colors from the set of available colors (i.e., nodes need to pick a color that has not yet been assigned to a neighbor) and in both cases, it is shown that in each such coloring round, every uncolored node becomes colored with constant probability. This directly implies that the node-averaged complexity of those algorithms is $O(1)$. Further, in a recent paper, Ghaffari and Kuhn~\cite{GK21} give a deterministic distributed algorithm to compute a $(\Delta+1)$-coloring (and more generally a $(\mathit{degree}+1)$-list coloring) in $O(\log^2\Delta\cdot\log n)$ rounds. The core of the algorithm is a method to color a constant fraction of the nodes of a graph in $O(\log^2\Delta+\log^* n)$ rounds, which directly implies that the deterministic node-averaged complexity of $(\Delta+1)$-coloring (and $(\mathit{degree}+1)$-list coloring) is $O(\log^2\Delta + \log^* n)$. Note that an improvement to this bound for $(\mathit{degree}+1)$-list coloring would immediately also improve the best known deterministic worst-case complexity.

In addition, most modern randomized distributed graph algorithms are based on the idea of \emph{graph shattering}, see, e.g., \cite{barenboim2016locality, harris2016distributed, ghaffari2016MIS,fischer2017sublogarithmic, GS17, chang2018optimal, ghaffari2018derandomizing}. In a first shattering phase, one uses a randomized algorithm that succeeds at each node with probability at least $1-1/\poly\Delta$ such that (essentially), the graph induced by the unsolved nodes consists of only small components. Those components are then solved in a second post-shattering phase by using the best known deterministic algorithm to complete the given partial solution. Since the shattering phase solves the given problem for all, except at most a $(1/\poly\Delta)$ fraction of all the nodes, the complexity of the shattering phase is typically an upper bound on the node-averaged complexity. Usually, the round complexity of the shattering phase is expressed as a function of $\Delta$ rather than as a function of $n$ and it is often also much faster than the deterministic post-shattering phase. The randomized sinkless orientation algorithm of ~\cite{GS17} for example  implies that the randomized node-averaged complexity of computing a sinkless orientation is $O(1)$.

\subsection{Organization of the Paper}
In \Cref{sec:model}, we formally define the notions of node and edge averaged complexities that we use in this paper. In \Cref{sec:algorithms}, present our upper bounds on the node and edge averaged complexities of maximal matching, $(2,2)$-ruling sets, and sinkless orientation. For space reasons, some of the arguments in \Cref{sec:algorithms} are only sketched and the full proofs are deferred to \Cref{app:missingAlgProofs}. In \Cref{sec:LB}, we then provide our lower bounds and thus in particular our lower bound on the node-averaged complexity of MIS in general graphs and on the worst-case complexity of MIS in trees. Also in \Cref{sec:LB}, some of the technical arguments appear in the appendix, in \Cref{apx:isomorphism}.

\section{Model and Definitions}
\label{sec:model}

We primarily focus on the \LOCAL model~\cite{linial1987LOCAL,peleg00}, where a network is modeled as an undirected graph $G=(V,E)$ and typically, every node $v\in V$ is equipped with a unique $O(\log n)$-bit identifier. Time is divided into synchronous rounds: In every round, every node of $G$ can send an arbitrary message to each neighbor and receive the messages sent to it by the neighbors. In the closely related \CONGEST model~\cite{peleg00}, messages are required to consist of at most $O(\log n)$ bits.

We study graph problems in which upon terminating, each node and/or
each edge of a graph $G=(V,E)$ must compute some output. We consider
the individual node and edge complexities of a distributed algorithm for a given
graph problem. In particular, we are interested in the time required
for individual nodes or edges to compute and commit to their outputs. For a node
$v\in V$, we say that $v$ has completed its computation as soon as $v$
and all its incident edges have committed to their outputs and we say
that an edge $e=\set{u,v}\in E$ has completed its computation as soon
as $e$ and both its nodes $u$ and $v$ have committed to their
outputs. For example, in a vertex coloring algorithm, a node $v$ has completed its computation as soon as $v$'s color is fixed and an edge $e=\set{u,v}$ has completed its computation as soon as the colors of $u$ and $v$ are fixed. In an edge coloring algorithm, a node $v$ has completed its computation as soon as the colors of all incident edges have been determined and an edge $e$ has completed its computation as soon as the color of $e$ is fixed. Given a distributed algorithm $\calA$ on $G=(V,E)$ and a node
$v\in V$, we define $T_v^G(\calA)$ to be the number of rounds after
which $v$ completes its computation when running the algorithm
$\calA$. Similarly, for an edge $e\in E$, we define $T_e^G(\calA)$ to
be the number of rounds after which $e$ completes its computation. Note that if $\calA$ is a randomized algorithm, then
$T_v^G(\calA)$ and $T_e^G(\calA)$ are random variables. In the
following, for convenience, we generally define $T_v^G(\calA)$ and
$T_e^G(\calA)$ as random variables. In the case of a
deterministic algorithm $\calA$, the variables only take on one
specific value (with probability $1$). We can now define the node and edge averaged complexities of a distributed algorithm $\calA$.

\begin{definition}[Node and Edge Averaged Complexities]\label{def:avg_complexity}
  We define the following average complexity measures for a
  distributed algorithm $\calA$ on a family of graphs $\calG$. We
  define the \emph{node-averaged complexity} ($\nodeavg(\calA)$) and the
  \emph{edge-averaged complexity} ($\edgeavg(\calA)$) as follows.
  \begin{eqnarray*}
    \nodeavg(\calA) & := & \max_{G\in \calG} \,\frac{1}{|V|}\cdot\E\left[\sum_{v\in V(G)}T_v^G(\calA)\right]\ =\
    \max_{G\in \calG} \frac{1}{|V|}\cdot\sum_{v\in V(G)}\E\big[T_v^G(\calA)\big]\\
    \edgeavg(\calA) & := & \max_{G\in \calG} \,\frac{1}{|E|}\cdot\E\left[\sum_{e\in E(G)}T_e^G(\calA)\right]\ =\
    \max_{G\in \calG} \frac{1}{|E|}\cdot\sum_{e\in E(G)}\E\big[T_e^G(\calA)\big]
  \end{eqnarray*}
\end{definition}

The respective complexity of a given graph problem is defined as the corresponding complexity, minimized over all algorithms $\calA$ that solve the given graph problem. We note that there are other complexity notions that are between \Cref{def:avg_complexity} and the standard worst-case complexity notion. We provide a brief discussion of this in \Cref{sec:other}.

\paragraph{Computation vs.\ Termination Time:}  After completing the computation, a node $v$ or edge $e$ might still be involved in communication to help other nodes determine their outputs. We say that a node $v$ has terminated once it has completed its computation and it also does not send any further messages. Similarly, an edge $e$ has terminated once it has completed its computation and there is no more messages sent over the edge. Instead of defining $T_v^G(\calA)$ and $T_e^G(\calA)$ as the number of rounds until $v$ or $e$ finishes its computation, we could also define it as the number of rounds until $v$ or $e$ terminates. In fact, in the literature about averaged complexity of distributed algorithms, both definitions have been used. The initial work by Feuilloley~\cite{feuilloley2020long} uses the definition that we use in the present paper. The subsequent work by Barenboim and Tzur~\cite{barenboim2019distributed} uses the stronger definition, where the complexity of a node/edge is defined as its termination time.

From a practical point of view, both notions of averaged complexity seem relevant. On the one hand, once a node has computed its output, it can continue with any further computation that is based on this output, even if the node still has to continue communicating to help other nodes. On the other hand, node-averaged termination time might be more natural especially if we aim to minimize the total energy spent in a distributed system. From a purely combinatorial / graph-theoretic point of view, when using the \LOCAL model, the node computation time definition has a particularly clean interpretation: An $r$-round \LOCAL algorithm can always equivalently be seen as an algorithm, where every node first collects its complete $r$-hop neighborhood and it then computes its output as a function of this information.\footnote{In the case of randomized algorithms, we have to assume that nodes choose all private random bits at the beginning before sending the first message.} More generally, if the computation time of a node $v$ is $r$, it means that node $v$ can compute its output as a function of its $r$-hop neighborhood in the graph and the node-averaged complexity in the \LOCAL model is therefore equal to the average radius to which the nodes must know the graph in order to compute their outputs. We note that while for algorithms, a termination time bound is stronger than an equal computation time bound, the opposite is true for lower bounds, and the definition we use therefore makes our lower bounds stronger. Further, although we use computation time in our definition, for all our algorithms, it is not hard to see that they also provide the same bounds if using average termination time instead of average computation time. In all our algorithms, nodes also stop participating in the algorithm at most one round after knowing their outputs.

\section{Algorithms}
\label{sec:algorithms}

\subsection{MIS and Ruling Set}
\paragraph{MIS} It is well-known that Luby's randomized MIS algorithm removes $1/2$ of the edges, per iteration~\cite{luby86}. Hence, if we define the MIS problem as a labeling problem, with binary indicator labels for vertices indicating those that are in the selected MIS, and we declare an edge terminated when the label of at least one of its two endpoint nodes is fixed, then Luby's algorithm has edge-averaged complexity $O(1)$.\footnote{Note however that if we use the edge-averaged complexity as defined in \Cref{def:avg_complexity} and require both nodes of an edge to be decided, then this is not true. In fact, in this case, we prove a  lower bound on the edge-averaged complexity in \Cref{thm:complexityMIS}.} In contrast, the node-averaged complexity of MIS had remained elusive for a number of years, and indeed it was mentioned as an open question throughout the literature whether an $O(1)$ node-averaged complexity is possible, see e.g., \cite{chatterjee2020sleeping}. The best known upper bounds are the trivial $O(\log n)$ that follows from Luby's worst-case round complexity analysis and a $O(\log \Delta/\log\log \Delta)$ bound that follows from the work of Bar-Yehuda, Censor-Hillel, Ghaffari, and Schwartzman~\cite{bar2017distributed}. They give a randomized MIS algorithm for which they show that within this time, each node is removed with at least a constant probability (and indeed a better probability that exceeds $1-1/\sqrt{\Delta}$, cf.\ Theorem 3.1 in \cite{bar2017distributed}).

In \Cref{sec:LB}, we show that the node-averaged complexity of MIS cannot be $O(1)$, and indeed the above bound is tight for small $\Delta$. Concretely, we prove that in a certain graph family with maximum degree $\Delta$ and $n$ nodes, the node-averaged complexity of MIS is $\Omega\Big(\min\Big\{\frac{\log \Delta}{\log\log \Delta}, \sqrt{\frac{\log n}{\log\log n}}\Big\}\Big)$. That is, asymptotically the same bounds as the celebrated worst-case lower bound Kuhn, Moscibroda, and Wattenhofer~\cite{kuhn16_jacm} also hold for the node-averaged complexity.

\paragraph{Ruling Set.} Faced with the above strong lower bound for the node-averaged complexity of MIS, it is natural to ask whether any reasonable relaxation of the problem admits better node-averaged complexity. One of the most standard relaxations of MIS is ruling set. An $(\alpha, \beta)$-ruling set asks for a set of nodes $S$ such that any two nodes in $S$ have distance at least $\alpha$ and any node not in $S$ has a node in $S$ within distance $\beta$. Thus, MIS is equivalent to a $(2,1)$-ruling set. Interestingly, we show that the seemingly minimal relaxation to $(2,2)$-ruling set drops the node-averaged complexity to $O(1)$.

\begin{theorem}
    There is a randomized distributed algorithm in the \CONGEST model that computes a $(2,2)$-ruling set and has node-averaged complexity $O(1)$.
\end{theorem}
\begin{proof}
The algorithm works as follows. Each node $v$ independently marks itself with probability $p_v:=1/(\deg(v)+1)$. A marked node $v$ joins the ruling set $S$ if and only if it has no marked higher priority neighbor $w$. A neighbor $w$ of $v$ is higher priority if $\deg(w)>\deg(v)$ or if $\deg(w)=\deg(v)$ and $\mathrm{ID}(w) > \mathrm{ID}(v)$. Nodes that are within distance $2$ of nodes in $S$ are deleted and we recurse on the remaining graph.

To prove the theorem, we show that per iteration, in expectation, a constant fraction of the nodes is deleted. Fix one iteration: We call a node $v$ \emph{good} if $\sum_{u\in N_2^+(v)} p_u \geq 1/2$, where we define $N_2^+(v):=\{u\in V : d_G(u,v)\leq 2\}$. To show that in expectation a constant fraction of the nodes is deleted, we show that at least $1/2$ of the nodes are good, and moreover that each good node is deleted with a constant probability.

We next show that at least half of the nodes are good. Let $B$ be the set of bad nodes (i.e., the set of nodes that are not good). To upper bound $|B|$, we do the following for each node $v\in B$. We charge the ``badness" of $v$ to the neighbors of $v$ by assigning value $1/\deg(v)$ to each neighbor $u$ of $v$. Because $v$ is bad, we have $\sum_{w\in N(u)} p_w < 1/2$ for each neighbor $u$ of $v$. Each node $u$ therefore gets charged less than $1/2$ by neighboring bad nodes. Because every bad node distributes a total charge of $1$, this means that at most half of the nodes can be bad.

To finish the proof, we show that a good node $v$ is deleted with constant probability. Let $N_2^+(v)$ be the (inclusive) $2$-hop neighborhood of $v$. We need to show that with constant probability at least one node in $N_2^+(v)$ joins the ruling set $S$. First, note that with constant probability at least one node in $N_2^+(v)$ is marked. This is simply because $v$ is good. Now if $M$ is the set of marked nodes in $N_2^+(v)$, we let $M'\subseteq M$ be the set of nodes in $M$ which have no higher priority neighbor in $M$. Note that if $|M|>0$, then also $|M'|>0$. Now, assume that $|M'|>0$ and consider some node $u\in M'$. Node $u$ enters the ruling set unless a higher priority neighbor $w$  outside $N_2^+(v)$ is marked. There are at most $\deg(u)$ such neighbors and because they must have degree $\deg(w)\geq \deg(u)$, each of them is marked with probability $\leq 1/(\deg(u)+1)$. The probability that none of the higher priority neighbors of $u$ is marked is thus at least a constant.
\end{proof}

\begin{theorem}\label{thm:RulingDet}
    There is a deterministic distributed algorithm in the \CONGEST model that computes a $(2, O(\log \Delta))$-ruling set in any $n$-node graph with maximum degree $\Delta$ and has node-averaged complexity $O(\log^* n)$. The algorithm can also be modified to produce a $(2, O(\log \log n))$-ruling set with the same $O(\log^* n)$ node-averaged complexity.
\end{theorem}
\begin{proof}[Proof Sketch]
The proof is deferred to \Cref{app:missingAlgProofs}. The basic idea is to apply $O(\log \Delta)$ iterations as follows: we use a simple dominating set $S$ algorithm that runs in $O(\log^* n)$ rounds and computes a dominating set of size at most $n/2$. All nodes outside $S$ point to $S$ and terminate, and we continue to the next iteration with only nodes $S$. After $O(\log \Delta)$ iterations, at most $n/\Delta$ nodes remain and we can compute an MIS of them in $O(\Delta + \log^* n)$ time using known algorithms~\cite{BEK15}. To get a $(2, O(\log \log n))$-ruling set, we stop after only $O(\log \log n)$ iterations, when the number of remaining nodes has dropped to $n/\poly(\log n)$, and we invoke a $\poly(\log n)$-round MIS algorithm~\cite{rozhonghaffari20} among the remaining nodes.
\end{proof}

\subsection{Maximal Matching}
\begin{theorem}\label[theorem]{thm:matchingRand}
    There is a randomized distributed algorithm in the \CONGEST model that computes a maximal matching, which has edge-averaged complexity $O(1)$, and a worst-case round complexity of $O(\log n)$, with high probability.
\end{theorem}
\begin{proof}[Proof Sketch] The result readily follows from a variant of Luby's algorithm~\cite{luby86}, adapted for maximal matching. The full proof is deferred to \Cref{app:missingAlgProofs}. The algorithm is to mark each edge $e=\{u, v\}$ with probability $\frac{1}{4(d_v+ d_u)}$ and to then add marked edges to the maximal matching if no other incident edge is marked. We then remove the matched vertices and repeat in the remaining graph. We show that per iteration, in expectation a constant fraction of the edges get removed. This part is somewhat analogous to the analysis of Luby's MIS algorithm~\cite{luby86}. In particular, we call each node $v$, with degree $d_v$, \emph{good} if at least $1/3$ of its neighbors have degree at most $d_v$, we show that at least $1/2$ of the edges are incident to good nodes, and each good node is matched with at least a positive constant probability. Hence, for at least $1/2$ of edges, each of them has a constant probability of being removed, which means in expectation a constant fraction of the edges gets removed. See \Cref{app:missingAlgProofs} for the proof.
\end{proof}
We comment that the above statement is also implicit in the classical maximal matching result of Israeli and Itai \cite{IsraelI86}\footnote{We thank an anonymous PODC'22 reviewer for bringing this to our attention.}. 

\begin{theorem}
    There is a deterministic \CONGEST algorithm that computes a maximal matching and has edge-averaged complexity $O(\log^2 \Delta + \log^* n)$, node-averaged complexity $O(\log^3 \Delta + \log^* n)$, and worst-case complexity $O(\log^2 \Delta \cdot \log n)$.
\end{theorem}
\begin{proof}
Let us describe one iteration of the algorithm, which takes $O(\log^2 \Delta + \log^* n)$ rounds and removes a constant fraction of the nodes. 
Let $E$ be the set of all edges in this iteration and consider the fractional matching where we assign to each edge $e=\{v, u\}\in E$ the fractional value $f_e=\frac{1}{d_v+d_u}$. Notice that this is indeed a valid fractional matching in the sense that for each node $v$ we have $\sum_{e \,\textit{ s.t.}\, v\in e} f_{e} \leq d_{v} \cdot \frac{1}{d_v} =1 $. Let us define for each edge $e=\{v, u\}$ a weight $w_e=(d_v+d_u)$. Then, the aforementioned fractional matching has total weight $\sum_{e \in E} f_e \cdot w_e = \sum_{e \in E} 1 = |E|.$ Using the deterministic rounding algorithm of Ahmadi, Kuhn, and Oshman~\cite{ahmadi2018distributed} for weighted matchings, we can compute an integral matching in $O(\log^2 \Delta + \log^* n)$ rounds whose weight is at least a constant factor of the fractional matching that we start with. That is, we get an integral matching with weight $|E|/10$. We add all the edges of this matching to our output maximal matching, and we then remove all the edges incident to matched nodes, and continue to the next iteration. Notice that for each edge $e=\{v, u\}$ that gets added to the matching, we can say it killed the $d_v + d_v-1 \geq \frac{d_v+d_u}{2} = w_e/2$ edges that share an endpoint with $e$, and this way each edge is killed at most twice, once by each endpoint. Hence, for any integral matching with weight $W$, removing its matched vertices removes at least $W/4$ edges from the graph. Therefore, with our rounding of the fractional matching that had weight $|E|$, we have found an integral matching whose addition removes at least $|E|/40$ edges. Since per iteration we spend $O(\log^2 \Delta)$ rounds and remove a $1/40$ fraction of the edges, we conclude that the edge-averaged complexity is $O(\log^2 \Delta + \log^* n)$.

Note that the $O(\log^* n)$ term is not needed in each repetition and it suffices to have it only in the first repetition. More concretely, in the algorithm of \cite{ahmadi2018distributed}, this term changes to an $O(\log^* \Delta)$ term if we are already given a $\poly(\Delta)$ coloring of the nodes \cite{ahmadi2018distributed} and we can compute that initially before all the repetitions in $O(\log^* n)$ rounds using Linial's classic algorithm\cite{linial1987LOCAL}. Hence, after having spent this $O(\log^* n)$ rounds at the start, each repetition takes $O(\log^2 \Delta + \log^* \Delta) = O(\log^2 \Delta)$ rounds and removes $1/40$ fraction of the edges. This directly also shows that after $O(\log^2 \Delta \cdot\log n + \log^* n) = O(\log^2 \Delta \cdot\log n)$ rounds, all the edges are removed and thus the algorithm has terminated. 

We next argue about the node-averaged complexity. If we repeat the algorithm for $\Theta(\log \Delta)$  iterations, for a total round complexity of $O(\log^3 \Delta + \log^* n)$, then the total number of edges in the graph is reduced by a factor $(1-1/40)^{\Theta(\log \Delta)} \leq 1/(2\Delta)$. Hence, the total number of remaining nodes (which must have degree at least one) is decreased by at least a factor $2$. The reason is that if we had $n$ nodes before, we had at most $n\Delta/2$ edges, and after the reduction of edges by a $1/(2\Delta)$ factor, the number of remaining edges is at most $n/4$ and thus we have at most $n/2$ nodes of degree at least $1$, i.e., at least $n/2$ nodes have degree $0$ and are thus removed. 
This means that in $O(\log^3 \Delta + \log^* n)$ rounds, the number of nodes reduces by a factor $1/2$ and thus the node-averaged complexity is $O(\log^3 \Delta + \log^* n)$.   
\end{proof}

\subsection{Sinkless Orientation}
The randomized sinkless orientation algorithm of Ghaffari and Su~\cite{GS17} already has node-averaged complexity  $O(1)$.\footnote{This statement is directly correct for the algorithm that they provide for graphs with a minimum degree of at least $500$. We believe that their extension to graphs with min-degree in $[3, 500]$ can also be adapted to have this $O(1)$ node-averaged complexity, basically by replacing their deterministic ruling set subroutine with a randomized one.} We next show a deterministic algorithm that achieves an $O(\log^* n)$ node-averaged complexity. Note that the worst-case complexity of this problem has an $\Omega(\log n)$ lower bound even in $3$-regular graphs~\cite{brandt2016lower}.
\begin{theorem}\label{thm:sinkless}
    There is a deterministic distributed \LOCAL model algorithm to compute a sinkless orientation of any $n$-node graph with minimum degree $3$, with node-averaged complexity $O(\log^* n)$ and worst-case complexity $O(\log n)$.
\end{theorem}
\begin{proof}[Proof Sketch]
We sketch the high-level idea for the special class of a high-girth graph $G$ of constant maximum degree. As a first step, we compute an MIS $S$ of $G^{2r+1}$ for a sufficiently large constant $r$ and use this set $S$ to cluster the graph. This produces clusters of diameter $O(r)$ such that the complete $r$-hop neighborhood of any node $v\in S$ is contained in its cluster. Note that this clustering can be computed in $O(\log^* n)$ rounds. We then build a virtual graph between the cluster centers (i.e., the nodes in $S$) such that each node is connected to three other cluster centers through disjoint paths, where those paths form the edges of this virtual graph. This is possible if the graph has a minimum degree $3$ and the girth is at least $c\cdot r$ for a sufficiently large constant $c$. If we compute a sinkless orientation of this virtual graph, we can obtain a sinkless orientation of $G$ by orienting the paths according to the sinkless orientation on the virtual graph and orienting all other edges in $G$ towards the nodes that participate in the virtual graph. In this way, we essentially reduce the problem of computing a sinkless orientation on $G$ to computing a sinkless orientation on a virtual graph with $|S|\lesssim n/3^r$ and where communication between neighbors costs $O(r)$ rounds on $G$. All nodes that are not part of the virtual graph are decided after $O(\log^* n)$ rounds. If the constant $r$ is chosen large enough, $O(r)\ll 3^r$ and we can keep the node-averaged complexity in $O(\log^* n)$. A full proof that also works for general graphs appears in \Cref{app:missingAlgProofs}.
\end{proof}

\section{Lower Bounds}
\label{sec:LB}
In this section, we prove lower bounds for the MIS problem. More precisely, we show a lower bound on the average complexity of MIS. As a key technical tool, we modify the KMW construction \cite{kuhn16_jacm} and obtain also a lower bound on the worst-case complexity of MIS, that holds already on trees.
In fact, the original KMW lower bound applies to the problem of finding a good approximation for the minimum vertex cover, and through a chain of reductions, it is shown that this implies a lower bound for MIS on line graphs. We modify the KMW lower bound construction to directly provide a lower bound for MIS on trees, that holds even for randomized algorithms, and then we show that this result also implies a lower bound for the average complexity of the MIS problem. The technical aspects of our proof follow the simpler version of the proof of the KMW lower bound shown by Coupette and Lenzen \cite{breezing}.

\subsection{Summary of the KMW Lower Bound}
On a high level, the KMW lower bound is obtained by showing that there exists a family of graphs satisfying that:
\begin{itemize}[noitemsep]
    \item There are two sets of nodes $S_0$ and $S_1$ that have the same view up to some distance $k$;
    \item Both $S_0$ and $S_1$ are independent sets;
    \item Every node of $S_0$ has $1$ neighbor in $S_1$, and every node of $S_1$ has $\beta$ neighbors in $S_0$, for some parameter $\beta$;
    \item $S_0$ is much larger than $S_1$, and it contains the majority of the nodes of the graph.
\end{itemize}
In these graphs, since $S_0$ is an independent set that contains the majority of the nodes, if we want to cover all edges, we could just select all nodes except for the ones of $S_0$, and obtain a very small solution to the vertex cover problem. But since nodes of $S_0$ and $S_1$ have the same view up to distance $k$, then, they either spend more than $k$ rounds to understand in which set they are, or they must have the same probability of joining the vertex cover. This probability of being selected should be at least $1/2$ if we want every edge between $S_0$ and $S_1$ to be covered, implying that any algorithm running in at most $k$ rounds fails to produce a small solution, because in expectation $|S_0|/2$ nodes of $S_0$ get selected.

More in detail, in order to obtain a lower bound for vertex cover approximations, \cite{kuhn16_jacm} follows this approach:\\
    \textbf{Define Cluster Tree Skeletons.} These trees define some properties that each graph in the family should satisfy, and they are parametrized by a value $k$. In general, given a cluster tree $CT_k$, there could be many graphs $G_k$ that satisfy the properties required by $CT_k$. The required properties are the following. 
    \begin{itemize}[noitemsep]
        \item Each node of a cluster tree $CT_k$ corresponds to an independent set of nodes in the graph $G_k$.
        \item Each edge $(v_1,v_2)$ in this tree is labeled with two values, $x$ and $y$, that dictate that nodes of the sets corresponding to $v_1$ and $v_2$ in the graph must be connected with a $(x,y)$-biregular graph.
    \end{itemize} 
    \textbf{Show Isomorphism.} Cluster trees are defined in a specific way that allows proving the following. Let $CT_k$ be a cluster tree, and let $G_k$ be a graph derived from $CT_k$. It is shown that, there are two ``special'' nodes in $CT_k$ that correspond to two ``special'' clusters $S_0$ and $S_1$ of $G_k$, such that, if $G_k$ has girth at least $2k+1$, then nodes of $S_0$ and $S_1$ have isomorphic views, up to distance $k$.\\
    \textbf{High-Girth Graphs.} It is shown that it is indeed possible to build a graph $G_k$ that satisfies the requirements of $CT_k$ and that has girth of at least $2k+1$. On a high level, this is shown by first constructing a low girth graph $G_k$, and then taking a high-girth lift of it.\\
    \textbf{Obtain Lower Bounds.} It is then shown that having the same view up to distance $k$ implies that, for a randomized algorithm running in graphs where IDs are assigned uniformly at random, nodes of $S_0$ and $S_1$ must have the same probability of joining a vertex cover, implying that many nodes of $S_0$ (which is shown to contain the majority of the nodes of $G_k$) must join the solution, while there exists a very small solution, not containing any node of $S_0$.
    
For a more detailed summary of the KMW lower bound, we refer the reader to Section 1.1 of \cite{breezing}---in this paper, Coupette and Lenzen provide a new and easier proof for the KMW lower bound.

\subsection{Our Plan}
In this work, we define our cluster tree skeletons in a very similar way as in Kuhn, Moscibroda, and Wattenhofer \cite{kuhn16_jacm}, but with a small and essential difference. Also, we slightly change the properties that a graph derived from a cluster tree should satisfy. On a high level, in our construction, every node of $CT_k$ corresponds to a set of nodes of $G_k$, which, differently from the original construction, is not independent anymore, with the only exception being the ``special'' set $S_0$ of nodes, which remains an independent set. 
We then show that, in each cluster of $G_k$ containing nodes that are neighbors of nodes of $S_0$, no large independent set exists, and that our construction is such that at least half of the nodes of $S_0$ must join the independent set in any solution for the MIS problem. In this way, we obtain that any algorithm running in $k$ rounds, on the one hand, cannot produce a large independent set in $S_1$, because it simply does not exist, and on the other hand, in order to guarantee the maximality constraint, it must produce a large independent set in $S_0$. If nodes of $S_0$ and $S_1$ have the same view up to distance $k$, the above is a contradicting requirement, preventing any $k$-round algorithm to exist.

As already mentioned before, the technical aspect of our proof is heavily based on the simplified version of the proof of the KMW lower bound shown by Coupette and Lenzen \cite{breezing}. In fact, we follow the same approach to prove the isomorphism between nodes of $S_0$ and $S_1$. 
We then deviate from this proof in order to show that high-girth graphs satisfying the required properties exist: we need to make sure that, in each cluster of $G_k$ containing nodes that are neighbors of nodes of $S_0$, no large independent set exists, while also making sure that the girth is at least $2k+1$. Actually, we do not achieve exactly this result: we first build a low girth graph satisfying the independence requirements, and then we make use of known graph-theoretical results \cite{ALM02} to lift it into a graph that also satisfies the independence requirements, and such that each node has a large probability of not seeing any cycle within distance $k$. Hence, the obtained graph may not have a high girth, but we prove that this weaker property is sufficient for our purposes.

\subsection{Cluster Trees}
\label{sec:clustertrees}

A cluster tree skeleton ($CT$) is a tree parametrized by an integer $k$, that we use to compactly describe a family of graphs $\mathcal{G}_k$ called cluster tree graphs. We now show how $CT_k$ is defined. We will later present the family of graphs described by a cluster tree skeleton $CT_k$. 

\begin{figure}[t]
	\centering
	\includegraphics[width=0.80\textwidth]{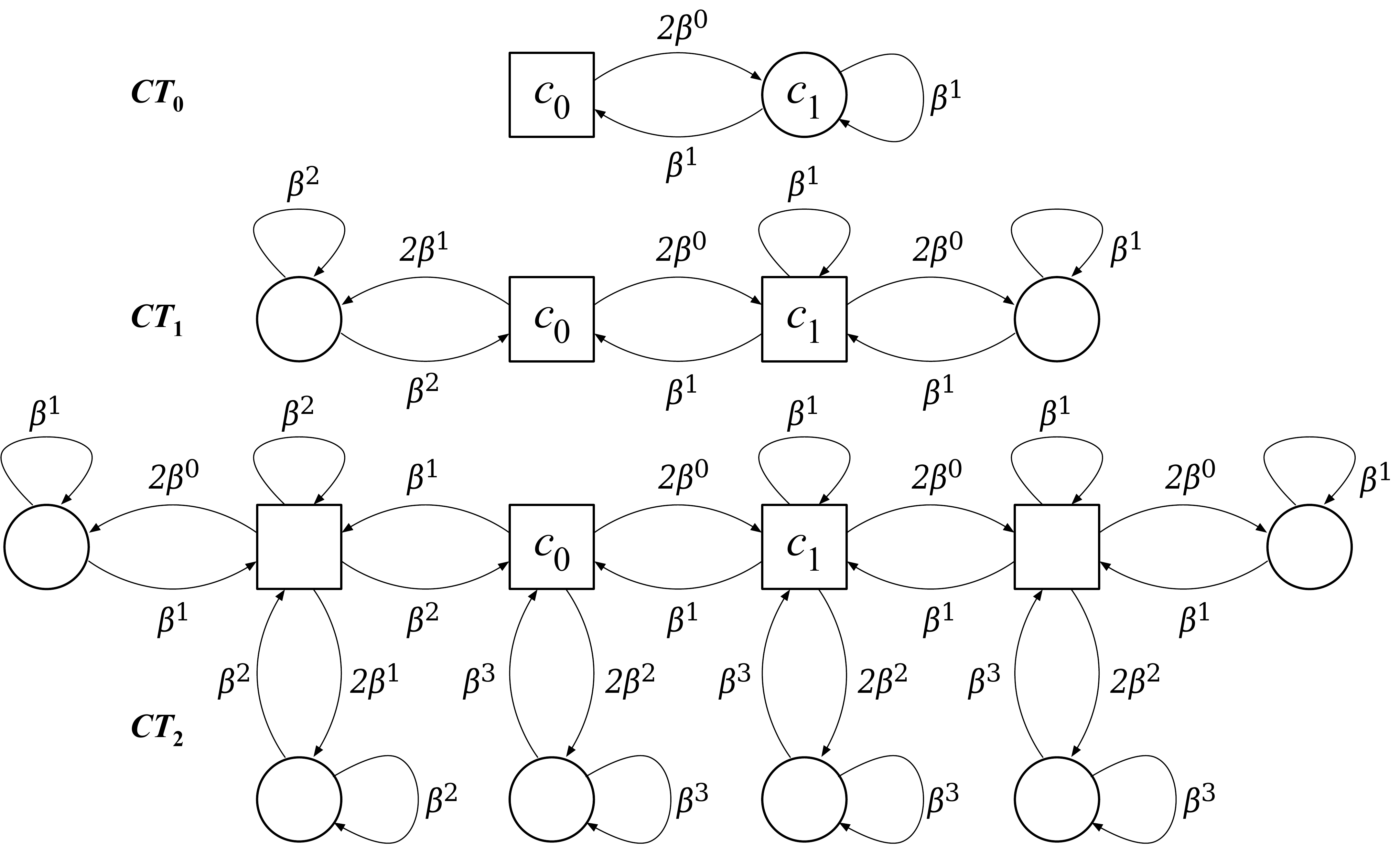}
	\caption{Examples of our cluster tree skeleton for $k\in\{0,1,2\}$; the squares represent internal nodes while circles represent leaves.} 
	\label{fig:cluster-trees}
\end{figure}

\paragraph{The Cluster Tree Skeleton.}
Differently from \cite{breezing}, our cluster tree skeletons $CT_k$ contain self loops, hence they are trees only when ignoring self loops. In fact, in our definition of $CT_k$, every node except one, will have a self loop. We denote with $(u,v,x)$ a directed edge $(u,v)$ labeled with some parameter $x$.
Given a parameter $\beta$, the cluster tree skeleton $CT_k$ is defined inductively as follows (see \Cref{fig:cluster-trees} for an example).
\begin{itemize}
    \item Base case: $CT_0 = (V_0,E_0)$, where $V_0 = \set{c_0,c_1}$, $E_0 = \set{(c_0,c_1,2\beta^0),(c_1,c_0,\beta^1),(c_1,c_1,\beta^1)}$. The node $c_0$ is called \emph{internal}, while $c_1$ is a \emph{leaf}. The node $c_0$ is the \emph{parent} of $c_1$, while $c_0$ has no parent.
    
    \item Inductive step: given $CT_{k-1}$, we define $CT_{k}$ as follows. 
         To each internal node $v$ of $CT_{k-1}$, we connect a leaf $\ell$ by adding the edges $(v,\ell,2\beta^{k})$ and $(\ell,v,\beta^{k+1})$. Moreover, we add the edge $(\ell,\ell,\beta^{k+1})$.
        Then, let $u$ be a leaf of $CT_{k-1}$ that is connected to its parent $p(u)$ through the edge $(u,p(u),\beta^i)$. We connect to each such node $u$ a node $\ell_j$ for each $j \in  \set{0,\ldots,k} \setminus \set{i}$ by adding the edges $(u,\ell_j,2\beta^j)$ and $(\ell_j,u,\beta^{j+1})$. Moreover, we add the edge $(\ell_j,\ell_j,\beta^{j+1})$. Node $u$ is now \emph{internal} in $CT_k$, and it is the \emph{parent} of all the added \emph{leaves} $\ell_j$.
\end{itemize}

\paragraph{The Graph Family.}\label{sec:graphfamily}
Given $CT_k = (V_k, E_k)$, the graphs $G_k \in \mathcal{G}_k$ are all the ones satisfying the following:
\begin{itemize}
    \item For each node $u \in V_k$, there is a set of nodes $S(u)$ in $G_k$;
    \item Let $u,v\in V_k$ be two nodes of $CT_k$, and let $S(u)$, $S(v)$ be two set of nodes in $G_k$ that represent $u$ and $v$, respectively. Also, let $(u,v,x) \in E_k$ be the directed edge from $u$ to $v$ labeled with a parameter $x$. Then, in $G_k$, all nodes in $S(u)$ must have exactly $x$ neighbors in $S(v)$.
    \item There are no additional edges in $G_k$.
\end{itemize}

\paragraph{Observations on $CT_k$ and $G_k$.}
Note that $CT_k$ is defined such that, if two (different) nodes are connected through an edge $(u,v,x)$, then there also exists some edge $(v,u,y)$ for some $y$. This implies that, by fixing the size of a single set $S(u)$, the size of all the other sets, and the maximum degree of the graph, are determined by the labels of the edges of $CT_k$. However, the exact way to connect nodes of different sets is not prescribed by the structure of $CT_k$, and hence there is freedom in realizing those connections. We will later show that it is possible to construct $G_k$ such that most of the nodes do not see cycles within distance $k$, while the maximum degree and the total number of nodes is not too large. 

We now observe some properties on the structure of $CT_k$.
\begin{observation}[Structure of $CT_k$]\label{obs:ctk}
    Each node of $CT_k$ is either internal or leaf. 
\begin{enumerate}[noitemsep]
    \item Every node $v \neq c_0$ has an edge $(v,v,\beta^i)$ for some $i$. We define $\psi(v) = i$, that is, $\psi(v)$ represents the exponent of the self loop of $v$.
    \item Each node $v$, except for $c_0$, has a parent $p(v)$, and has edges $(v,p(v),\beta^{i+1})$, $(p(v),v,2\beta^{i})$, $(v,v,\beta^{i+1})$, for some $i$ that satisfies $0 \le i \le k-1$ for internal nodes, and  $0 \le i \le k$ for leaves. 
    
    \item Let $v \neq c_0$ be an internal node, and let $i = \psi(v)$. Node $v$ has $k$ children that can be reached with edges $(v,u_j,2\beta^j)$ for all $j \in \set{0,\ldots,k} \setminus \set{i}$.
    
    \item Node $c_0$ has $k+1$ children that can be reached with edges $(c_0,u_j,2\beta^j)$ for all $j \in \set{0,\ldots,k}$.
\end{enumerate}
\end{observation}
We denote with $S^{-1}(v)$, for $v \in V(G_k)$, the node $v' \in V_k$ satisfying $v \in S(v')$.
We now label the edges $(u,v)$ of $G_k$, with a label that depends on the label of the edge $(u',v')$ of $CT_k$ from which $(u,v)$ comes from. 
\begin{definition}[Labeling of the Edges of $G_k$]\label{def:labeling}
    For every edge $(u,v)$ of $G_k$, let $u' = S^{-1}(u)$, and $v' = S^{-1}(v)$. Let $(u',v',x)$ be the edge of $CT_k$ connecting $u'$ and $v'$. We mark the edge $(u,v)$ with $\beta^i$ if $x = \beta^i$ or $x = 2 \beta^i$. Similarly as in the case of $CT_k$, we may refer to this edge of $G_k$ with $(u,v,\beta^i)$. If $u' = v'$, we additionally mark each edge with the label $\mathsf{self}$.
\end{definition}
\begin{observation}[Number of Neighbors with a Specific Label in $G_k$]\label{obs:numberofedges}
    Every node of $G_k$ that corresponds to an internal node of $CT_k$ has exactly $2 \beta^i$ outgoing edges marked $\beta^i$, for all $i \in \set{0,\ldots,k}$.
    Every node of $G_k$ that corresponds to a leaf of $CT_k$ gets exactly $2\beta^i$ outgoing edges labeled $\beta^i$, for exactly one $i \in \set{0,\ldots,k+1}$.
\end{observation}
\begin{observation}\label{obs:noescape}
    Let $v$ be a node of $G_k$ that corresponds to an internal node of $CT_d$. Let $(v,u,\beta^i)$ be an arbitrary edge incident to $v$. Then $u$ corresponds to a node in $CT_d$ if and only if $i \le d$.
\end{observation}

\subsection{Isomorphism Assuming No Short Cycles}
We now prove that for all graphs $G_k \in \mathcal{G}_k$, for all pair of nodes $(u,v) \in S(c_0) \times S(c_1)$, if their radius-$k$ neighborhood does not contain any cycle, then $u$ and $v$ have the same radius-$k$ view.
\begin{theorem}[$k$-hop Indistinguishability of Nodes Corresponding to $c_0$ and $c_1$]\label{thm:sameview}
Let $G_k \in \mathcal{G}_k$ be a cluster tree graph, and consider two nodes $v_0 \in S(c_0)$ and $v_1 \in S(c_1)$ that satisfy that $G^k_k(v_0)$ and $G^k_k(v_1)$ are trees. Then $v_0$ and $v_1$ have the same view up to distance $k$.
\end{theorem}
The proof is deferred to \Cref{apx:isomorphism}. We will follow the same approach of \cite{breezing}, that is, we provide an algorithm that defines an isomorphism, and then we prove that the algorithm is correct.

\subsection{Lifting}

We next show how to obtain a graph $\tilde{G}_k \in \mathcal{G}_k$ such that most nodes in $\tilde{G}_k$ see no short cycles and such that the graph induced by any of the clusters (except cluster $S(c_0)$) has no large independent set. The high-level idea of this construction is as follows. We start from some base graph $G_k\in\mathcal{G}_k$ and then we show that $\tilde{G}_k$ can be computed as a random lift of $G_k$. We first introduce the necessary graph-theoretic terminology.

For a graph $G=(V_G,E_G)$ and a graph $H=(V_H,E_H)$, a covering map $\varphi$ from $G$ to $H$ is a graph homomorphism $\varphi:V(G)\to V(H)$ such that for every node $v\in V(G)$, the neighbors of $v$ are bijectively mapped to the neighbors of $\varphi(v)$. That is, if $v$ has $d$ neighbors $u_1,\dots,u_d$ in $G$, then $\varphi(v)$ has $d$ (different) neighbors $\varphi(u_1),\dots,\varphi(u_d)$ in $H$. For a graph $G$, a graph $\tilde{G}=(\tilde{V},\tilde{E})$ for which a covering map from $\tilde{G}$ to $G$ exists is called a \emph{lift} of $G$. We will refer to the set of preimages $\varphi^{-1}(v)$ of a node $v\in V$ as the \emph{fiber} of $v$ in $\tilde{G}$. If $G$ is connected, then all fibers have the same cardinality, i.e., for every $v\in V$, there is the same number of nodes $\tilde{v}\in\tilde{V}$ for which $\varphi(\tilde{v})=v$. Even if $G$ is not connected, a lift $\tilde{G}$ of $G$ is typically constructed such that all fibers have the same cardinality. If all fibers of a lift $\tilde{G}$ of $G$ have the same cardinality $q\geq 1$, then $q$ is called the \emph{order} of the lift $\tilde{G}$. Note that if we have a graph $G_k\in \mathcal{G}_k$, then any lift $\tilde{G}_k$ of $G_k$ is also a graph in $\mathcal{G}_k$.

There are different ways to construct lifts with desirable properties for a given base graph. As discussed, our lower bound proof is based on the minimum vertex cover lower bound of \cite{kuhn16_jacm}. In \cite{kuhn16_jacm}, the lift $\tilde{G}_k$ of a base graph $G_k\in \mathcal{G}_k$ (for the family of graphs used in \cite{kuhn16_jacm}) is computed in a deterministic way such that $\tilde{G}_k$ has girth at least $2k+1$ (i.e., such that any $k$-hop neighborhood is tree-like). We cannot use the same deterministic lift construction because we also need to ensure that the induced graphs of the neighboring clusters of $S(c_0)$ in $\tilde{G}_k$ have no large independent sets. We instead use a random lift construction that was described in \cite{ALM02}. Given a base graph $G=(V,E)$ and an integer $q\geq 1$, we can obtain a random lift $\tilde{G}=(\tilde{V},\tilde{E})$ of order $q$ of $G$ as follows. For every $v\in V$, $\tilde{V}$ contains $q$ nodes $\tilde{v}_1,\dots,\tilde{v}_q$ and for every edge $\set{u,v}\in E$, the two fibers $\set{\tilde{u}_1,\dots,\tilde{u}_q}$ and $\set{\tilde{v}_1,\dots,\tilde{v}_q}$ are connected by a uniformly random perfect matching (which is chosen independently for different edges). We obtain the following, the proof of which is deferred to \Cref{apx:isomorphism}.

\begin{lemma}\label{lemma:randomlift}
    Let $G=(V,E)$ be a base graph with maximum degree $\Delta\geq 2$, let $q\geq 1$ be a positive integer and let $\tilde{G}=(\tilde{V},\tilde{E})$ be a random lift of order $q$ of $G$ as described above. Then, for every integer $\ell\geq 3$ and every node $\tilde{v}\in \tilde{V}$, the probability that $\tilde{v}$ is contained in a cycle of length at most $\ell$ is upper bounded by $\Delta^{\ell}/q$. Further, for a set of nodes $C\subseteq V$, let $\tilde{C}\subseteq \tilde{V}$ be the set of nodes consisting of the union of all fibers of the nodes in $C$. If $G[C]$ is a complete graph of size $|C|\geq 2$, then for every integer $s\geq 8\ln|C|$, $\alpha(\tilde{G}[\tilde{C}])\leq s\cdot q$ with probability at least $1-e^{-s^2q/8}$, where $\alpha(G)$ is the independence number of the graph $G$.
\end{lemma}

\subsection{(Almost) High-Girth Graphs in \boldmath$\mathcal{G}_k$}
We prove that there are graphs that belong to $\mathcal{G}_k$ such that, for each node, the probability that it sees a cycle within distance $k$ is small. We first show how to construct a low girth graph $G_k \in \mathcal{G}_k$, and then we use \Cref{lemma:randomlift} to show that we can use it to construct graphs $G'_k \in \mathcal{G}_k$ that satisfy the above. 
\paragraph{A Base Graph.}
We now construct a family of low girth graph $G_k$, parametrized by an even integer $\beta$, and we prove that these graphs belong to $\mathcal{G}_k$. The graph $G_k$, as a function of $\beta$, is defined as follows.
\begin{itemize}
    \item Let $d(v)$ be the hop distance of a node $v$ of $CT_k$ from $v_0$. Note that $0 \le d(v) \le k+1$.
 For each $v \neq v_0$, let $i=\psi(v) \le k+1$, be the exponent of the label of its self loop. We define $S(v)$ to be a set of nodes of size $z = 2 \beta^{k+1} (\beta/2)^{k+1-d(v)}$ nodes, connected as follows. Start from $t = z / \beta^{i}$ disjoint cliques of size $\beta^{i}$, $\set{C_1,\ldots,C_t}$. Note that $t$ is an even integer. For every $1 \le j \le t/2$, add to $G_k$ the edges of a perfect matching between the nodes of $C_j$ and $C_{t/2+j}$. In this way, we obtain that every node in $S(v)$ has exactly $\beta^{i}$ neighbors inside $S(v)$, and hence the requirements of the graph family defined in \Cref{sec:graphfamily} are satisfied for the edge $(v,v,\beta^i)$.
    \item $S(v_0)$ contains  $z = 2 \beta^{k+1} (\beta/2)^{k+1}$ nodes, and they form an independent set.
    \item Let $v$ be a node at distance $j$ from $v_0$, and $u$ be a neighbor of $v$ at distance $j+1$ from $v_0$. By construction, they are connected to each other with edges $(u,v,2\beta^i)$ and $(v,u,\beta^{i+1})$, for some $0 \le i \le k$. Observe that $S(v)$ has size  $2 \beta^{k+1} (\beta/2)^{k+1-j}$ and $S(u)$ has size $2 \beta^{k+1} (\beta/2)^{k-j}$. 
    We group the nodes of $S(v)$ into groups of size $\beta^{i+1}$, and the nodes of $S(u)$ into groups of size $2\beta^{i}$. We obtain that $S(v)$ and $S(u)$ are both split into $t = 2 \beta^{k+1} (\beta/2)^{k+1-j} / \beta^{i+1} = 2 \beta^{k+1} (\beta/2)^{k-j} / (2\beta^i)$ groups. Note that $t$ is an integer. We take a perfect matching between the groups, and for each matched pair of groups, we connect its nodes as $K_{\beta^{i+1},2\beta^i}$. We obtain that every node in $S(v)$ has $2\beta^i$ neighbors in $S(u)$, and every node in $S(u)$ has $\beta^{i+1}$ neighbors in $S(v)$, and hence the requirements of the graph family defined in \Cref{sec:graphfamily} are satisfied for the edges $(u,v,2\beta^i)$ and $(v,u,\beta^{i+1})$.
\end{itemize}

We summarize our result as follows. The proofs of the following statements appear in \Cref{app:missingHighGirth}
\begin{lemma}\label{lem:smallgirth}
    For all integers $k$ and even integers $\beta$ satisfying $2(k+1)/\beta < 1/2$, there exists a graph $G_k \in \mathcal{G}_k$, satisfying:
    \begin{itemize}[noitemsep]
        \item For all $v \neq v_0$, let $i=\psi(v)$. The graph $G_k$ satisfies that $\alpha(G_k[S(v)]) \le |S(v)| / \beta^i$.
        \item The total number of nodes is $O(\beta^{2k+2})$ and the maximum degree is at most $2\beta^{k+1}$.
    \end{itemize}
\end{lemma}

\paragraph{Applying the Lifting.}
We apply \Cref{lemma:randomlift} to the graph of \Cref{lem:smallgirth} to produce an almost high-girth version of it.
\begin{lemma}\label{lem:lifted}
    For any integer $q \ge 1$, $k$ and even $\beta$ satisfying $2(k+1)/\beta < 1/2$ there is a graph $\tilde{G}_k \in \mathcal{G}_k$ satisfying that:
    \begin{itemize}[noitemsep]
        \item The number of nodes is $O(q \beta^{2k+2})$ and the maximum degree is at most $2\beta^{k+1}$.
        \item For all $v$, for all $\ell \ge 3$, the probability that $v$ is contained in a cycle of length at most $\ell$ is upper bounded by $\Delta^\ell / q$.
        \item For all $v \in N(v_0)$, let $i=\psi(v)$, and let $t = 2 \beta^{k-i+1} (\beta/2)^{k}$. The graph $\tilde{G}_k$ satisfies that for every integer $s \ge 8 \ln \beta^i$, $\alpha(G_k[S(v)]) \le s q t$ with probability at least $(1-e^{-s^2 q /8})^{t}$.
    \end{itemize}
\end{lemma}

We now fix the parameters $q$ and $\beta$ to obtain a friendlier version of \Cref{lem:lifted}.
\begin{corollary}\label{cor:family}
    Let $k$ be an integer and and $\beta \ge 4$ be an even integer satisfying $2(k+1)/\beta < 1/2$. For infinitely many values of $n$, there is a graph $\tilde{G}_k$ in $\mathcal{G}_k$ satisfying the following:
    \begin{itemize}[noitemsep]
        \item The number of nodes is $n = \beta^{O(k^2)}$ and the maximum degree is at most $\Delta = 2\beta^{k+1}$.
        \item The probability that a node $v$ is contained in a cycle of length at most $2k+1$ is at most $1/\beta$.
        \item For each $v \in N(v_0)$ satisfying that $i=\psi(v)$, $\alpha(\tilde{G}_k[S(v)])\le \gamma |S(v)| \frac{\log \beta^i}{\beta^i}$ for some constant $\gamma > 0$.
    \end{itemize}
\end{corollary}

\subsection{Lower Bounds for MIS}
In the remaining, we prove the following theorem.
\begin{theorem}\label{thm:complexityMIS}
The randomized node and edge averaged complexities of the MIS problem in general graphs and the randomized worst-case complexity of MIS on trees are both $\Omega\Big(\min\Big\{\frac{\log\Delta}{\log\log\Delta},\sqrt{\frac{\log n}{\log\log n}}\Big\}\Big)$.
\end{theorem}
\begin{proof}
Consider the family of graphs described in \Cref{cor:family}, for an arbitrary parameter $k$, and $\beta = \Omega(k^2 \log k)$. Let $v_1,\ldots, v_{k+1}$ be the neighbors of $v_0$, where $v_i$ is the one reached with the edge $(v_0,v_i,2\beta^{i-1})$. Note that $\psi(v_i)=i$. Let $S_i = S(v_i)$. For each $S_i$ it holds that at most $s_i = \gamma |S_i| \frac{\log \beta^i}{\beta^i} $ nodes can join the independent set, and these nodes can cover at most $s_i \beta^{i} = \gamma |S_i| i \log \beta = \gamma 2i|S_0| \log \beta /\beta$ nodes of $S_0$, where the last equality holds because the ratio between the sizes of $S_i$ and $S_0$ is $\beta/2$. Hence, in total, nodes in $S_1,\ldots,S_{k+1}$ cover at most $|S_0| \cdot 2 \gamma \frac{(k+1)^2}{\beta} \log \beta$ nodes of $S_0$. For $\beta = \Omega( k^2 \log k)$, we obtain that at most $|S_0|/2$ nodes of $S_0$ are covered. Note that nodes of $S_0$ cannot cover other nodes of $S_0$ (since $S_0$ is an independent set), and hence if a node of $S_0$ is not covered by any node of $S_1,\ldots,S_{k+1}$ then it must join the independent set. Hence, at least $|S_0|/2$ nodes of $S_0$ must join the independent set.

Every node of $S_0$ has probability at least $1 - 1/\beta$ to see, in $k$ rounds, a tree-like neighborhood, and in that case, by \Cref{thm:sameview}, they have the same view of those nodes of $S_1$ that also see a tree-like neighborhood within distance $k$, and since those nodes must have probability at most $\gamma\log \beta / \beta$ of joining the independent set, then we obtain that, in expectation, a fraction $1 - 1/\beta$ of the nodes of $S_0$ have probability at most $\gamma\log \beta / \beta$ of joining the independent set. Even if all the nodes of $S_0$ that do not have a tree-like neighborhood join the independent set, we still obtain that, in expectation, at most a fraction $(1/ \beta + \gamma\log \beta / \beta)$ of nodes of $S_0$ join the independent set, which is $< 1/2$ (for a large enough $\beta$), implying that the obtained independent set is not maximal.

Hence, at least half of the nodes of $S_0$, which is the majority of the nodes of the graph, cannot decide within their first $k$ rounds, if they join the MIS or not, implying a lower bound of $\Omega(k)$ on the node-averaged complexity of the MIS problem. For the edge-averaged complexity, note that a $(1-o(1))$-fraction of the edges are incident to nodes in $S_0$ and for those edges to be decided, their nodes in $S_0$ must be decided. The half of $S_0$ that cannot decide within the first $k$ rounds covers half of the edges incident to nodes in $S_0$ and therefore, also the edge-averaged complexity of MIS is $\Omega(k)$.

Let us now focus on trees, and determine a lower bound on the worst-case randomized complexity of MIS. Let us consider a tree obtained by taking the radius-$k$ view of a node $v'_0$ in $S_0$ that does not see any cycle within distance $k$, such that all its neighbors do not see any cycle within their distance $k-1$. We then add some nodes at distance at least $k$ from $v'_0$ to make the total number of nodes the same as in $G_k$, matching also the maximum degree in $G_k$. In this way, we obtain a tree, where $v'_0$ has the same $(k-1)$-radius neighborhood of nodes in $S_0$ of $\tilde{G}_k$ that do not see any cycle within distance $k-1$. Moreover, node $v'_0$ has $k+1$ neighbors $v'_1,\ldots,v'_{k+1}$ satisfying that $v'_i$ has the same $(k-1)$-radius neighborhood of nodes in $S_i$ that do not see any cycle within distance $k-1$.
We claim that, for any randomized algorithm for trees that runs in at most $k-1$ rounds, for all $1 \le i \le k+1$, $v'_i$ must have probability $p_i$ of joining the MIS that is at most a factor $1 /(1-1/\beta)$ larger than the one of nodes in $S_i$ that do not see any cycle within distance $k$.
In fact, assume that $v'_i$ joins the MIS with probability larger than $p_i$. 
If we simulate the $(k-1)$-rounds algorithm for trees on the nodes of $\tilde{G}_k$ that do not see any cycle within distance $k$, which are, in expectation, a fraction $(1-1/\beta)$ of the total nodes, then we obtain that the expected number of nodes in $S_i$ joining the MIS is larger than $\alpha(\tilde{G}_k[S_i])$, which cannot happen. Also, since $v'_0$ has the same $(k-1)$-radius neighborhood of $v'_1$, then $p_0 = p_1$.
Hence, the $(k-1)$-algorithm for trees, by a union bound, covers $v'_0$ with probability at most 
\[
\sum_{i=0}^{k+1} p_i \le \frac{1}{1-1/\beta} \cdot (k+2)\frac{\log \beta}{\beta} \le \frac{1}{k},
\]
where the first inequality holds because nodes in $S_i$ have probability of joining the MIS at most $\gamma \log \beta^i / \beta^i$, which is maximized for $i=1$, and the last inequality holds because $\beta = \Omega(k^2 \log k)$. Hence, the algorithm fails to cover the node $v'_0$ with high probability, implying a lower bound of $\Omega(k)$ on trees as well.

Finally, we argue that an $\Omega(k)$ lower bound implies the lower bounds claimed in the theorem. Since the degree is $\Delta = 2\beta^{k+1}$, by taking $\beta = \Theta( k^2 \log k)$, we obtain 
$ \Delta =  \Theta((k^2 \log k)^{k+1}) = k^{\Theta(k)}$,
and hence $\Omega(k) = \Omega(\log \Delta / \log \log \Delta)$. Then, the total number of nodes is
$ n = \beta^{O(k^2)} = k^{O(k^2)}$,
and we thus get that $k$ is as large as $\Omega(\sqrt{\log n / \log \log n})$.
\end{proof}

\subsection{Lower Bound on Node-Averaged Complexity of Maximal Matching}
\begin{theorem}\label{thm:lb_avgMatching}
The randomized node-averaged complexity of the maximal matching problem in general graphs is $\Omega\Big(\min\Big\{\frac{\log\Delta}{\log\log\Delta},\sqrt{\frac{\log n}{\log\log n}}\Big\}\Big)$.
\end{theorem}
\begin{proof}
    The proof follows almost immediately from the maximum lower bound construction in \cite{kuhn16_jacm}, the details are deferred to \Cref{app:lb_avgMatching}.
\end{proof}

\urlstyle{same}
\bibliographystyle{alpha}
\bibliography{ref}

\appendix

\section{Other Notions of Averaged Complexity}
\label{sec:other}

We conclude this paper by discussing possible stronger notions of node and edge averaged complexity. In some cases, we might want to have stronger guarantees. For example, if we want to use an algorithm with a low node or edge-averaged complexity as a building block in another algorithm, the individual node or edge complexities might be particularly important for some of the nodes or edges and not as important for others. In some cases, we might even need all nodes or edges to have a low complexity ``on average", when averaging over different runs of the algorithm. For this purpose, we also define the \emph{node/edge weighted averaged complexity} and the \emph{node/edge expected complexity}. Recall that $T_v^G(\calA)$ and $T_e^G(\calA)$ are the individual complexities of node $v$ and edge $e$ when running algorithm $\calA$ on graph $G$. For the node weighted averaged complexity, the nodes in $V$ are given positive weights $w: V\to\mathbb{R}_{>0}$, which are given to an algorithm as input. The weighted node-averaged complexity of an algorithm is then defined as the weighted average node complexities according to the given weights. The weighted edge-averaged complexity is defined accordingly with edge weights. For a given problem $\calP$, the weighted node/edge-averaged complexity is defined as the weighted node/edge-averaged complexity of best possible algorithm, but for a worst-case weight distribution. 

We want to point out that the concept of weighted node-averaged complexity has (implicitly) proven useful in a recent distributed graph coloring paper by Ghaffari and Kuhn~\cite{GK21}. For given weights on the nodes, they give an efficient algorithm to color a subset of nodes corresponding to a constant fraction of the overall weight and they use this to implement some pipelining, when coloring layered graphs. As a result, they obtained faster deterministic algorithms for computing $\Delta$-colorings of general graphs and for coloring low-arboricity graphs.

The node expected complexity of an algorithm $\calA$ on a graph $G=(V,E)$ is defined as $\max_{v\in V}\E[T_v^G(\calA)]$ and again the edge expected complexity is defined accordingly. Note that for a deterministic algorithm $\calA$, the node/edge expected complexity is equal to the worst-case complexity. This is however not true for randomized algorithms, where it can be that every node/edge has a small individual complexity in expectation, but that with high probability there are some nodes/edges for which the individual complexity in an actual execution is significantly larger (this is in fact the case for many existing randomized graph algorithms that are based on graph shattering~\cite{barenboim2016locality, harris2016distributed, ghaffari2016MIS,fischer2017sublogarithmic, GS17, chang2018optimal, ghaffari2018derandomizing}). We note that Feuilloley~\cite{feuilloley2020long} defines a notion of averaged complexity that is essentially equivalent to the node expected complexity as defined in this paper. Feuilloley considers the maximal (over the nodes) individual complexity of the best possible deterministic algorithm if the node identifiers are chosen uniformly at random.

In the following, for a graph problem $\calP$, let us use $\nodeavg(\calP)$, $\wnodeavg(\calP)$, $\nodeexp(\calP)$, and $\nodeworst(\calP)$ to denote the node-averaged complexity of $\calP$, the weighted node-averaged complexity of $\calP$, the node expected complexity of $\calP$, and the worst-case complexity of $\calP$. We have
\[
\nodeavg(\calP)\ \leq\ \wnodeavg(\calP)\ \leq\ \nodeexp(\calP)\ \leq\ \nodeworst(\calP).
\]
For deterministic algorithms, the last inequality is an equality. For the edge-based complexity notions, the definitions and relations hold accordingly. 

Let us briefly observe what we can say about the weighted averaged and expected complexities of MIS and maximal matching, two of the main problems studied in this paper. For MIS and for maximal matching, we showed that the node-averaged complexity is lower bounded by the KMW bound (i.e., by $\Omega\Big(\min\Big\{\frac{\log\Delta}{\log\log\Delta},\sqrt{\frac{\log n}{\log\log n}}\Big\}\Big)$). The same lower bound therefore also holds for the stronger node-averaged complexity notions. For maximal matching, the edge-averaged complexity is $O(1)$. Note however that, as described in the proof of \Cref{thm:lb_avgMatching}, the KMW lower bound construction that is used for maximum matching in \cite{kuhn16_jacm} contains a perfect matching such that every maximal matching contains almost all edges of this perfect matching, and within the KMW time bound, only a $o(1)$-fraction of those edges can join the matching. By putting most of the weight on those edges, we can therefore see that the KMW lower bound also holds for the weighted edge-averaged complexity of maximal matching (and thus also for the edge expected complexity). We mentioned in \Cref{sec:intro} that the node-averaged complexity of MIS is $O\big(\frac{\log\Delta}{\log\log\Delta}\big)$, by using a randomized MIS algorithm of \cite{bar2017distributed}. This algorithm in fact has the stronger property that for every node $v\in V$, the probability that $v$ is undecided after $O\big(\frac{\log\Delta}{\log\log\Delta}\big)$ rounds is $o(1)$. Cf.\ Theorem 3.1 in \cite{bar2017distributed}. This implies that the node expected complexity of MIS and the edge expected complexity of maximal matching are both $O\big(\frac{\log\Delta}{\log\log\Delta}\big)$. We hope that further investigating the discussed and possibly other notions of node or edge-averaged complexity will prove useful for obtaining a more thorough understanding of the complexities of distributed graph problems.

\section{Missing Proofs from Section 3}
\label{app:missingAlgProofs}

\begin{proof}[\bf Proof of \Cref{thm:matchingRand}]
Per iteration, we mark each edge $e=\{u, v\}$ with probability $\frac{1}{4(d_v+ d_u)}$ where $d_u$ and $d_v$ denote the degrees of $u$ and $v$. Then, any marked edge that is not incident on any other marked edge is added to the maximal matching, and we remove all nodes that have a matching edge (and all the edges incident on any such node). We then continue to the next iteration.

We use some of Luby's standard terminology and ideas~\cite{luby86}, though our analysis is somewhat different as we show that the number of edges remaining in the maximal matching problem shrinks by a constant factor\footnote{Notice that the counterpart statement of number of nodes remaining in the MIS problem shrinking by a constant factor is not correct, for Luby's algorithm, or any other MIS algorithm, as we formally prove in our lower bound.}. Let us call a node $v$ good if at least $d_v/3$ of the neighbors of $v$ have degree at most $d_v$. We call an edge good if it is incident on at least one good node.
We can use a standard argument to show that at least $1/2$ of the edges of the graph are good.\footnote{To keep the proof self-contained, here is the argument: direct each edge to the higher degree endpoint (breaking ties arbitrarily). For each bad node, the number of outgoing edges is greater than twice the number of incoming edges. Hence, each bad edge, which enters a node $v$, can be matched to a distinct set of two edges that exit the node $v$. Therefore, the total number of edges is at least $2$ times the number of bad edges.} We next show that and each good edge is removed with a constant probability.

Consider a good edge $e=\{u, v\}$ and suppose that $v$ is the (or one of the two) good endpoint(s). Let us examine the set of at least $d_v/3$ edges of node $v$ that go to neighbors with degree at most $d_v$. Each such edge $e'=\{v, w\}$ is marked with probability at least $\frac{1}{8d_{v}}$. Moreover, the probability that any other edge incident on $e'$ is marked is at most $\frac{d_{v}}{4(d_v)} + \frac{d_w}{4(d_w)} = 1/2$. Hence, with probability at least $\frac{1}{32d_{v}}$, the edge $e'$ is added to the matching. Since these events are mutually exclusive for the at least $d_v/3$ edge of $v$ going to lower or equal degree neighbors, we conclude that with probability at least $\frac{d_v}{3} \cdot \frac{1}{32 d_v} > \frac{1}{100}$, node $v$ gets a matching edge and hence edge $e$ gets removed.

Hence, per iteration, at least $1/2$ of edges are good and each good edge is removed with probability at least $1/100$. Thus, the edge-averaged complexity of the algorithm is $O(1)$. Moreover, it is easy to argue that the worst case complexity is $O(\log n)$ with high probability:  per iteration the expected number of edges shrinks by a constant factor. Thus, if we repeat the algorithm for $O(\log n)$ iterations, the number of remaining edges is expected to be below $\frac{1}{n^2}$. By Markov's inequality, this implies that with probability $1-1/n^{2}$ no edge remains and the algorithm has terminated.
\end{proof}

\begin{proof}[\bf Proof of \Cref{thm:RulingDet}]
We make use of a known\footnote{We can also get this result through the following very simple algorithm. First, every node picks one arbitrary outgoing edge. We then get an oriented pseudo-forest that contains all nodes. In this pseudo-forest, we add all the parents of the leaves to the MIS, and remove them along with their neighbors. Then, we augment the MIS by adding to it a maximal independent set of the nodes in the remaining pseudo-forest.} deterministic dominating set algorithm of Kutten and Peleg~\cite{kutten1998fast} that runs in $O(\log^* n)$ rounds and computes a dominating set of size at most $n/2$. To get a deterministic $(2,O(\log\Delta))$-ruling set algorithm with average complexity $O(\log^* n)$, we do as follows. At the beginning all nodes are active. We then run $O(\log\Delta)$ phases in which in each phase, we compute a dominating set $S$ of the graph induced by the active nodes such that $S$ contains at most half of the active nodes. All nodes not in $S$ point to a neighbor in $S$. At the end, we have at most $n/\Delta$ active nodes and we can then afford to compute an MIS of the remaining active nodes in time $O(\Delta + \log^* n)$ using known algorithms~\cite{BEK15}.

To compute a $(2,O(\log\log n))$-ruling set in deterministic average complexity $O(\log^* n)$. We stop, when we have $n/\mathrm{poly}\log n$ nodes and use the network decomposition based $\poly(\log n)$ round deterministic MIS algorithm of Rozhon and Ghaffari~\cite{rozhonghaffari20} to compute an MIS of the remaining nodes.
\end{proof}

\begin{proof}[\bf Proof of \Cref{thm:sinkless}]
We first describe the process for one iteration, and afterward discuss the repetitions and how they affect the overall (averaged) complexity. Let $r$ be a large enough constant to be fixed later. We call each cycle short if it has length at most $6r$. First, we perform a partial orientation to ensure that any node that has an edge in a short cycle has out-degree at least $1$. 

\paragraph{Partial orientation, edges in short cycles} For each short cycle $C$, take the edge $e_C$ with the smallest identifier in $C$ and define the \emph{preferred orientation} of $C$ as the orientation that goes in $e_C$ from the smaller identifier endpoint of $e_C$ to the other endpoint, and goes through the other edges of $C$ accordingly. We call an edge short if it is in at least one short cycle. We orient each short edge $e$ consistent with the preferred orientation of the smallest id short cycle that contains $e$, where the id of a short cycle is a concatenation of the ids of its edges (thus any two different cycles have different ids).

We next argue that with this orientation, each node that has a short edge gets out-degree at least $1$. Consider a node $v$ and let $C$ be the smallest id short cycle that contains one $v$. Suppose $C$ has edges $e_1 = \{v, u_1\}$ and
$e_2 = \{v, u_2\}$ incident on $v$. Both $e_1$ and $e_2$ are oriented according to the preferred orientation of $C$, as they are oriented according to the smallest id short cycle that contains them, and that is $C$ by its definition. Hence, either $e_1$ is oriented as $v \rightarrow u_1$, in which case we are done, or $e_1$ is oriented as $u_1 \rightarrow v$ but $e_2$ is oriented as $v \rightarrow u_2$. In either case, node $v$ has out-degree at least $1$. 

\paragraph{Nodes with no short edge, max-degree reduction.} Notice that what remains are nodes that do not have any edge in a short cycle. None of the edges of these nodes is oriented. To continue the work on these, we would like to switch to a graph where the maximum degree is also upper bounded by a constant. For that, let each such node choose $3$ of its edges arbitrarily. For any edge $e=\{v, u\}$, if the edge $e$ is chosen by only one endpoint $v$ (either because $u$ already had an outgoing edge or because it chose other edges), then we consider $e$ a self-loop incident on node $v$. In particular, node $u$ will not rely on this edge $e$ to have an outgoing edge, and node $v$ can safely assume that this edge will be oriented from $v$ to $u$. We also remove all nodes that had a short edge incident on them, and thus already have an outgoing edge. At this point, we have a graph on the remaining nodes where each node is incident on exactly $3$ edges or self-loops, all edges remain unoriented, and there is no cycle of length $6r$ in the graph (with the exception of self-loops).

\paragraph{Clustering of remaining nodes.} Let $S$ be the set of nodes that have distance at least $2r+1$ from any self-loop. Let $S'$ be a maximal $(2r+1)$-independent subset of $S$, which can be computed in $O(\log^* n)$ rounds as $r=O(1)$~\cite{linial1987LOCAL}. Let $S''$ be the union of $S'$ and the set of nodes that have a self-loop incident on them. We cluster all nodes where each node is clustered with the closest node in $S''$ (breaking ties by ID). Each node has within its distance $2r+1$ either a self-loop or a node from $S'$ (as otherwise it would be added to $S'$ and thus $S''$). Hence the radius of each cluster is at most $2r+1$. Moreover, any two clusters that include vertices that are adjacent have at most one edge connecting them, as two edges would imply a cycle of length at most $2(2r+1)+2 = 4r+4\leq 6r$.

For any cluster that is centered on a self-loop, orient all the edges with both endpoints in the cycle toward the self-loop. Any other cluster has radius at least $r$ as all nodes within distance $r$ of the cluster center are in its cluster (this is because nodes of $S'$ are $(2r+1)$-independent and have distance at least $2r+1$ from self-loops). For any of these other clusters centered at $S'$-nodes, let the cluster center choose three of its neighboring clusters arbitrarily: we then keep alive only the nodes on the paths to these three neighboring clusters. For all other nodes of the cluster, we orient the edges toward (the closest) one of these three paths. Hence, we keep alive at most $3(2r+1) \leq 7r$ nodes in the cluster, while the cluster had at least $3(2)^{r}$ vertices. 

\paragraph{Contraction, and repetitions.} To prepare for the next iteration, we now think of a contracted graph where each of these clusters and its three paths are simply one node and its three edges. Notice that since the cluster radius is at most $2r+1$, the round complexity of running any algorithm on this contracted graph is at most a $4r+4$ factor higher. However, this cost is only incurred by the alive nodes which are at most a $\frac{7r}{(2)^{r}}$ fraction of nodes. Hence, the total round complexity contribution of the next iteration is smaller by a factor of  $\frac{(7r) (4r+4)}{(2)^{r}} \leq \frac{1}{2}$, where the last inequality follows as $r\geq 15$. Hence, per iteration, the total added round complexity (added over all nodes that remain alive) shrinks by a $1/2$ factor. This implies that the total round complexity summation of all nodes is dominated by the first iteration and is thus $O(\log^* n)$.   

\paragraph{Overall analysis, and round complexity improvement} Repeating this algorithm, per iteration, the number of remaining nodes shrinks by a constant factor and thus all nodes would terminate within $O(\log n)$ repetitions. That would make for a total round complexity of $O(\log n \cdot \log^* n)$. To improve this to $O(\log n)$, after $O(\log\log n)$ iterations, by which time the number of remaining nodes has dropped below $\frac{n}{\log^2 n}$, we switch to the standard $O(\log n)$ algorithm to finish the remaining nodes. Since only $\frac{n}{\log^2 n}$ participate in that $O(\log n)$ round algorithm, the contribution to the node-averaged complexity is $o(1)$ and we can conclude that this modified algorithm still has node-averaged complexity of $O(\log^* n)$ and worst-case round complexity of $O(\log n)$.
\end{proof}

\section{Missing Proofs from Section \ref{sec:LB}}\label{apx:isomorphism}

\subsection{Isomorphism Assuming No Short Cycles}
\paragraph{The Algorithm.}
The algorithm used to prove the isomorphism is almost the same as the one used in \cite{breezing}, and it is shown in Algorithm \ref{algo:isomorphism}. The only difference between our algorithm and the one in \cite{breezing} lies in the lines \ref{code:nvi} and \ref{code:nwi}, which refer to the labels assigned to the edges of $G_k$ in \Cref{def:labeling}.
This will slightly alter the size of the lists $N_v[i]$ used in the algorithm, compared to those produced in the original case of \cite{breezing}, but we will prove that the algorithm still works as expected.

On a high level, the algorithm takes in input two nodes $v_0 \in S(c_0)$ and $v_1 \in S(c_1)$, and does the following. With $G^i(v)$ we denote the subgraph induced by nodes at distance at most $i$ from $v$ in $G$, excluding the edges between nodes at distance exactly $i$. Assume that the views of $v_0$ and $v_1$ are tree-like up to distance $k$, that is, $G^k_k(v_0)$ and $G^k_k(v_1)$ are trees. The algorithm visits these views recursively and builds an isomorphism $\phi$ between nodes in the view of $v_0$ and nodes in the view of $v_1$. In order to build such an isomorphism, since the views are tree-like, it is enough to pair nodes such that paired nodes have the same degree. Every time the algorithm recurses on two paired nodes $v$ and $w$, it groups the neighbors of $v$ (resp.\ $w$) into $k+2$ lists $N_v[i]$ (resp.\ $N_w[i]$), where in $N_v[i]$ (resp.\ $N_w[i]$) are contained nodes reachable through edges labeled $\beta^i$. If for all $i$, the size of the lists $N_v[i]$ and $N_w[i]$ are the same, then it is trivial to produce an isomorphism, but unfortunately, this is not always the case. However, we will prove that the sizes are almost always the same, and the function $\textsc{Map}$ takes care of handling the remaining cases.

We follow the same strategy of \cite{breezing} to prove that the algorithm is correct. We start with defining the \emph{history} of a node. Informally, the history of a node is the label of the edge (of $G_k$) connecting the current node to the one that we used to reach it.
\begin{definition}[Node History]
    Let $z \in \set{u,v}$, be a node that is reached by the algorithm through the call of the function $\textsc{Walk}(u,v,\mathrm{prev},\mathrm{depth})$ satisfying $\mathrm{prev} \neq \bot$. If $z=u$, let $p = \mathrm{prev}$, otherwise let $p = \phi(\mathrm{prev})$. Let $(z,p,\beta^i)$ be the edge connecting $z$ to $p$. The history of $z$ is $\beta^i$.
\end{definition}
Note that, in general, the history of a node may differ from the label used to reach that node. In fact, if $u$ is reached through the edge $(p,u,\beta^j)$, and $u$ can reach $p$ through the edge $(u,p,\beta^i)$, then $i=j$ if and only if $p$ is in the same cluster of $v$ and the edge $(p,u)$ corresponds to a self loop in $CT_k$, while otherwise $i$ is either $j+1$ or $j-1$.

In the following, we will use the terms \emph{internal} and \emph{leaf} when talking about nodes of $G_k$: a node is internal (resp.\ leaf) if it corresponds to an internal node (resp.\ leaf) of $CT_k$.
We prove that the size of the variables $N_v[i]$ and $N_w[i]$ only depends on the \emph{position} of a node (internal or leaf) and its history. 
\begin{lemma}[Variables Determining Node Neighborhoods]\label{lem:listsize}
    For every call of $\textsc{Map}(N_v,N_w)$ it holds that the content of $N_v$ and $N_w$ only depends on the position and the history of $v$ and $w$. In particular, if $v$ and $w$ have the same position and history, then $\mathrm{len}(N_v[i]) = \mathrm{len}(N_w[i])$ for all $i \in \set{0,\ldots,k+1}$. If $v$ and $w$ are both internal, but they have different histories $\beta^x$ and $\beta^y$, then $\mathrm{len}(N_v[i]) = \mathrm{len}(N_w[i])$ for all $i \in \set{0,\ldots,k+1} \setminus \set{x,y}$, $\mathrm{len}(N_v[x]) = \mathrm{len}(N_w[x]) - 1$, and $\mathrm{len}(N_v[y]) - 1 = \mathrm{len}(N_w[y])$.
\end{lemma}
\begin{proof}
    By \Cref{obs:numberofedges}, all internal nodes have exactly $2\beta^i$ outgoing edges labeled $\beta^i$, for each $i \in \set{0,\ldots,k}$, and $0$ edges labeled $\beta^{k+1}$. Hence, if two internal nodes have also the same history, then all the lists have the same size. If two internal nodes $v$ and $w$ have different histories $\beta^x$ and $\beta^y$, then $N_v[x]$ contains exactly one less element than $N_w[x]$, because the algorithm excludes only $\mathrm{prev}$ from being part of $N_v(x)$. Similarly, $N_w[y]$ contains exactly one less element than $N_v[y]$, because $\phi(\mathrm{prev})$ is not part of the list.
    
    By \Cref{obs:numberofedges}, leaf nodes that have the same history $\beta^x$ have exactly $2\beta^x$ outgoing edges labeled $\beta^x$, and they do not have any other outgoing edge with different labels. Hence, for these nodes, $\mathrm{len}(N_v[x]) = \mathrm{len}(N_w[x]) = 2 \beta^x - 1$ (where the $-1$ comes from the fact that $\mathrm{prev}$ and $\phi(\mathrm{prev}$) are excluded from the lists of $v$ and $w$, respectively), and all the other lists have length $0$.
\end{proof}

By using a characterization corresponding to the one of \Cref{lem:listsize}, authors in \cite{breezing} showed that, in order to prove that their version of Algorithm \ref{algo:isomorphism} is correct, it is sufficient to prove that similar requirements as in the statement of \Cref{cor:sufficient} are satisfied. We will follow the same proof idea and adapt it to our case. 
\begin{corollary}[Sufficient Conditions for the Correctness of Algorithm \ref{algo:isomorphism}]\label{cor:sufficient}
    Given a graph $G_k$ and two nodes $v_0 \in S(c_0)$ and $v_1 \in S(c_1)$ that satisfy that $G^k_k(v_0)$ and $G^k_k(v_1)$ are trees, if all pairs of nodes mapped by Algorithm \ref{algo:isomorphism} either have the same position and history, or they are both internal, then Algorithm \ref{algo:isomorphism} produces an isomorphism between $G^k_k(v_0)$ and $G^k_k(v_1)$.
\end{corollary}
\begin{proof}
    Since the $k$-radius neighborhoods of $v_0$ and $v_1$ are trees, then the visit produced by Algorithm \ref{algo:isomorphism} creates two rooted trees of depth $k$, and hence it is sufficient to show that, at every step, the degrees of the visited nodes are the same, and that the algorithm produces a bijection between the neighbors of the current nodes.
    
    We first show that $\mathrm{Map(N_{v_0},N_{v_1})}$ succeeds.
    Nodes $v_0$ and $v_1$ are both internal, and since $\mathrm{prev} = \bot$, by \Cref{obs:numberofedges} all lists have the same size, and the claim trivially holds.
    
    For nodes $v,w$ that have the same position and history, \Cref{lem:listsize} ensures that all lists have the same size, and hence $\mathrm{Map(N_v,N_w)}$ clearly succeeds.
    
    For nodes $v,w$ that have the same position internal but have different histories $\beta^x$ and $\beta^y$, by \Cref{lem:listsize}, all the lists at position different from $x$ and $y$ have the same size, and thus the function $\mathrm{Map(N_v,N_w)}$ clearly produces a bijection between nodes in these lists. Also, by \Cref{lem:listsize}, lists at position $x$ and $y$ only differ by $1$, and the function $\mathrm{Map(N_v,N_w)}$ maps the elements of $N_v[x]$ to the first $\mathrm{len}(N_w[x])-1$ elements of $N_w[x]$, the elements of $N_w[y]$ to the first $\mathrm{len}(N_v[y])-1$ elements of $N_v[y]$, and the last $2$ remained unmatched elements of $N_w[x]$ and $N_v[y]$ between each other. Hence, also in this case a bijection is produced.
\end{proof}

\begin{figure}[t]
	\centering
	\includegraphics[width=0.9\textwidth]{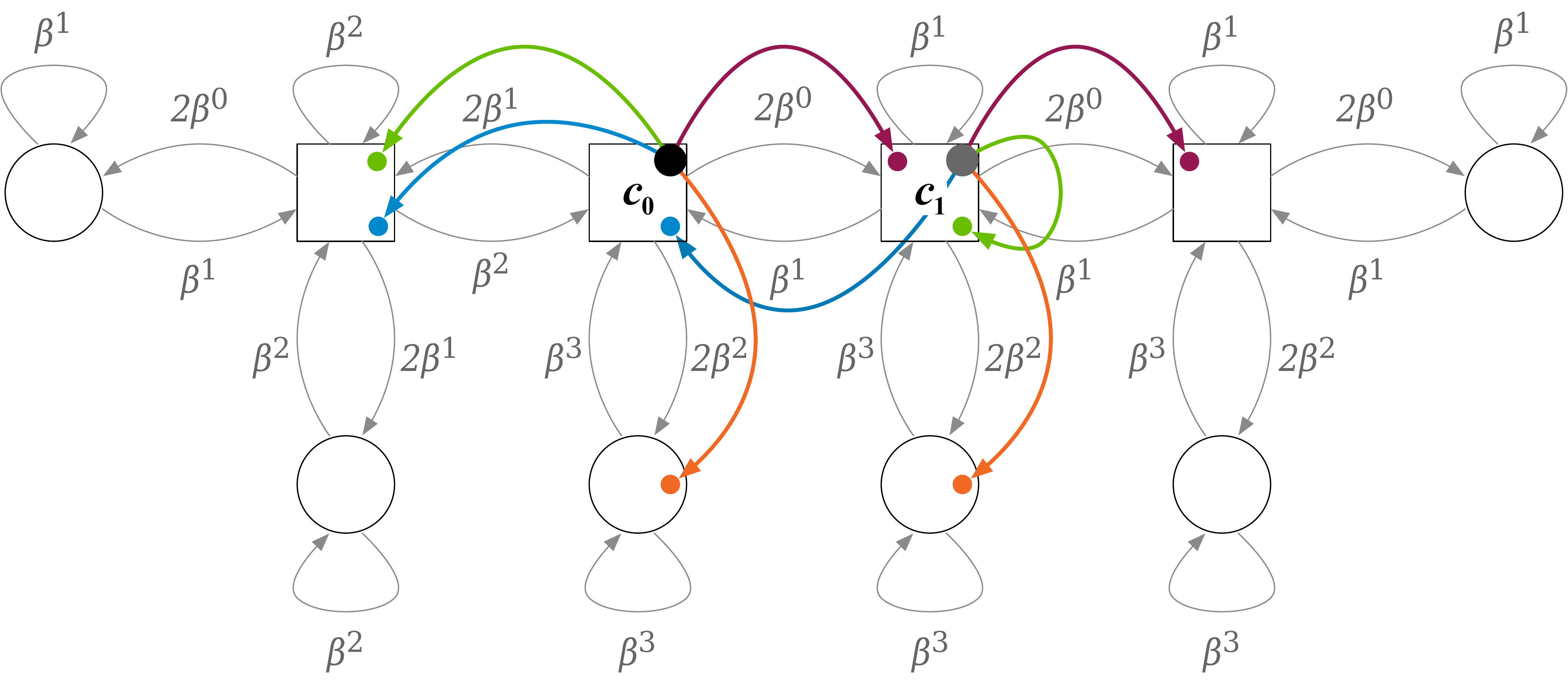}
	\caption{An example of the base case $d=1$ applied on $CT_2$. The black node in $c_0$ represents $v_0$, while the gray node in $c_1$ represents $v_1$.} 
	\label{fig:lem11-base}
\end{figure}

As in \cite{breezing}, we now prove an invariant satisfied by the algorithm. The invariant is exactly the same as the one proved in \cite{breezing}, while its proof differs. In fact, the additional thing that we need to handle is the case in which the algorithm follows an edge that corresponds to a self loop of $CT_k$, which turns out to be an easy case to handle. 
Recall that $V_k$ are the nodes of $CT_k$. 
\begin{lemma}[Main Invariant of Algorithm \ref{algo:isomorphism}]\label{lem:maininvariant}
    For $0 <  d < k$, suppose that $v$ and $w = \phi(v)$ lie at distance $d$ from $v_0$ and $v_1$, respectively. Let $v' = S^{-1}(v)$ and $w' = S^{-1}(w)$. Then exactly one of the following holds:
    \begin{enumerate}
        \item $v', w' \in V_d$, and $v$ and $w$ have the same history or both have history $\le \beta^{d+1}$.
        \item There is some $i$ with $d < i \le k$ such that $v', w' \in V_i \setminus V_{i-1}$, $v$ and $w$ have the same history, and $v'$ and $w'$ are connected to $CT_{i-1}$ with edges $(z,v',2\beta^j)$, $(v',z,\beta^{j+1})$, $(z',w',2\beta^j)$, $(w',z',\beta^{j+1})$, for some $z, z'$ and some $j \in \set{0,\ldots,i}$.
    \end{enumerate}
\end{lemma}
\begin{proof}
    We assume that $k$ is fixed, and we prove the claim by induction on $d$.
    \paragraph{Base Case ($d=1$).} 
    Let $v_0 \in S(c_0)$ and $v_1 \in S(c_1)$ be the two nodes on which the $\mathrm{FindIsomorphism}$ function is called, and let $v$ be an arbitrary neighbor of $v_0$ that is matched to $w$, a neighbor of $v_1$. That is, $w = \phi(v)$, the distance of $v$ from $v_0$ is $1$, and the distance of $w$ from $v_1$ is $1$.
    
    Since, by \Cref{obs:numberofedges} and the fact that $\mathrm{prev} = \bot$, for all $i$, $\mathrm{len}(N_{v_0}[i]) = \mathrm{len}(N_{v_1}[i])$, then only neighbors of $v_0$ and $v_1$ that are reachable from edges labeled the same are matched, that is, both $(v_0,v)$ and $(v_1,w)$ are labeled $\beta^j$ for the same $j$. 
    
    If $j = 0$, then, by the definition of cluster tree skeleton, and by the labels assigned in \Cref{def:labeling}, $v$ must be in $S(c_1)$, while $w$ must be in $S(\ell)$, where $\ell$ is the leaf attached to $c_1$ when constructing $CT_1=(V_1,E_1)$ (purple edges in \Cref{fig:lem11-base}). Thus both $v'$ and $w'$ are part of $V_1$, and moreover, by \Cref{obs:ctk} and \Cref{def:labeling}, $v$ and $w$ have the same history, hence the first case applies.
    
    If $j = 1$, then $v$ must be in $S(\ell)$, where $\ell$ is the leaf attached to $c_0$ when constructing $CT_1$ (green and blue edges that start from $v_0$ in \Cref{fig:lem11-base}), while $w$ must either be in $S(c_0)$ (by following an edge of $G_k$ corresponding to $(c_1,c_0,\beta^1)$; blue edge that starts from $v_1$ in \Cref{fig:lem11-base}), or in $S(c_1)$ (by following the self loop $(c_1,c_1,\beta^1)$; green edge that starts from $v_1$ in \Cref{fig:lem11-base}). In both cases, we always reach nodes of $G_k$ corresponding to nodes of $CT_1$, and hence $v'$ and $w'$ are both in $V_1$. Moreover, by the definition of cluster tree skeleton, all edges of $CT_1$ have parameter at most $2\beta^2$, which implies that both $v$ and $w$ have histories at most $\beta^2 = \beta^{d+1}$, and hence the first case applies.
    
    If $j > 1$, then $v$ and $w$ both lie in clusters $S(\ell_v)$ and $S(\ell_w)$ corresponding to nodes $\ell_v$ and $\ell_w$ that have been added as leaves when constructing $CT_j$ (orange edges in \Cref{fig:lem11-base}), which means that both $v'$ and $w'$ are part of $V_j \setminus V_{j-1}$. By \Cref{obs:ctk}, this implies the existence of the edges $(c_0,\ell_v,2\beta^{j})$, $(\ell_v,c_0,\beta^{j+1})$, $(c_1,\ell_w,2\beta^{j})$, $(\ell_w,c_1,\beta^{j+1})$ in $CT_k$, and hence, $v$ and $w$ also have the same history. The existence of the above edges and the fact that these nodes have the same history, implies that the second case applies.
    
    \paragraph{Inductive Case.} Assume that the claim holds for $0 < d < k-1$, we prove that the claim holds for $d+1$. Assume that $v$ and $w$ are such that $w = \phi(v)$, $v$ is at distance $d+1$ from $v_0$, and $w$ is at distance $d+1$ from $v_1$. Let $p(v)$ (resp.\ $p(w)$) be the node at distance $d$ from $v_0$ (resp.\ $v_1$) that is a neighbor of $v$ (resp.\ $w$), that is, the node that has been used to reach $v$ (resp.\ $w$). By applying the inductive hypothesis, we know that the claim holds for $p(v)$ and $p(w)$. Hence, we distinguish between the two possible cases.
    \begin{itemize}
        \item The first case of the invariant holds for $p(v)$ and $p(w)$, that is, $S^{-1}(p(v)),S^{-1}(p(w)) \in V_d$, and $p(v)$ and $p(w)$ have the same history, or both have history $\le \beta^{d+1}$. We consider two cases separately, either $S^{-1}(v) \in V_{d+1}$, or $S^{-1}(v) \not\in V_{d+1}$.
        \begin{itemize}
            \item {\boldmath $S^{-1}(v) \in V_{d+1}$}. Note that, by the definition of cluster tree skeleton, all edges that start from a node of $V_d$ and reach a node of $V_{d+1}$ are labeled with parameter $2\beta^j$ for some $j \le d+1$. Hence, by \Cref{def:labeling}, the edge $(p(v),v)$ is labeled with $\beta^j$ for some $j \le d+1$, implying that $v \in N_{p(v)}[j]$ for some $j \le d+1$. By \Cref{lem:listsize}, node $v$ is then matched either with a node in $N_{p(w)}[j]$, if $p(v)$ and $p(w)$ have the same history, or with a node in $N_{p(w)}[j']$ for $j' \le d+1$, if  $p(v)$ and $p(w)$ have different history but both $\le \beta^{d+1}$. Hence, $w \in N_{p(w)}[j']$ for some $j' \le d+1$, and since $S^{-1}(p(w))$ is internal in $CT_{d+1}$ because it already exists in $CT_{d}$ (by assumption), then by \Cref{obs:noescape}, $S^{-1}(w) \in V_{d+1}$. Then, since all edges of $CT_{d+1}$ are labeled either  $\beta^j$ or $2\beta^j$, for $j \le d+2$, then $v$ and $w$ have history at most $\beta^{d+2}$, and hence the first case applies.
            \item {\boldmath $S^{-1}(v) \not\in V_{d+1}$}. In this case, we have that $S^{-1}(v) \in V_{i} \setminus V_{i-1}$ for some $i > d+1$. Since $S^{-1}(p(v))$ is internal in $CT_{d+1}$ (because by assumption it already exists in $CT_d$), then the edge connecting $S^{-1}(p(v))$ to $S^{-1}(v)$ has been added to $CT_k$ when constructing $CT_i$ and adding $S^{-1}(v)$ as leaf, and hence the edge $(p(v),v)$ is labeled $\beta^i$, implying that $v \in N_{p(v)}[i]$. By \Cref{lem:listsize}, the length of the lists of $p(v)$ and $p(w)$ can only differ in positions $\le d+1$ (because, by assumption, either they have the same history, or they have history $\le \beta^{d+1}$), and since $i > d+1$, then $v$ is matched with a node in $N_{p(w)}[i]$, and hence $w \in N_{p(w)}[i]$. 
            Note that  $S^{-1}(p(w))$ is internal in $CT_{d+1}$, in $CT_i$, and in $CT_{i-1}$, because it already exists in $CT_{d}$ and $i-1 \ge d+1$. Then by \Cref{obs:noescape}, $S^{-1}(w) \in V_i$. Moreover, since $S^{-1}(p(w))$ is internal in $CT_{i-1}$ by following edges labeled $\beta^{i}$, by \Cref{obs:noescape}, we cannot reach nodes of $G_k$ corresponding to nodes of $V_{i-1}$, and hence $S^{-1}(w) \not\in V_{i-1}$. Hence, $S^{-1}(w) \in V_{i} \setminus V_{i-1}$, implying that it is a leaf in $CT_i$. This implies, by \Cref{obs:ctk} and \Cref{def:labeling}, that the edges $(v,p(v))$ and $(w,p(w))$ are both labeled $\beta^{i+1}$, and hence $v$ and $w$ have the same history.
            Hence, $S^{-1}(v)$ and $S^{-1}(w)$ are both connected as leaves when constructing $CT_i$, and they are connected with their parents with edges labeled the same. By \Cref{obs:ctk} the edges required by the second case of the invariant are present, for $j=i$. Thus, the second case of the invariant applies.
        \end{itemize}
        \item The second case of the invariant holds for $p(v)$ and $p(w)$, that is, let $v' = S^{-1}(p(v))$ and $w' = S^{-1}(p(w))$. There is some $i$ with $d < i \le k$ such that $v', w' \in V_i \setminus V_{i-1}$, $p(v)$ and $p(w)$ have the same history, and $v', w'$ are connected to $CT_{i-1}$ with edges $(z,v',2\beta^j)$, $(v',z,\beta^{j+1})$, $(z',w',2\beta^j)$, $(w',z',\beta^{j+1})$, for some $z, z'$ and some $j \in \set{0,\ldots,i}$. Hence, $v'$ and $w'$ are leaves in $CT_i$, and thus $p(v)$ and $p(w)$ not only have the same history, but also the same position. Notice that, the connection of $v'$ and $w'$ with $CT_{i-1}$ implies that they satisfy $\psi(v') = \psi(w') = j+1$, that is, they have the same self loop.
        Also, \Cref{lem:listsize} implies that $v$ and $w$ are contained, respectively, in $N_{p(v)}[j']$ and $N_{p(w)}[j']$ for the same $j'$, and that these two lists have the same length. Note that both $p(v)$ and $p(w)$ have edges labeled $\mathsf{self}$ in the list at position $j+1$, and nowhere else. If $j' = j+1$, then, note that the sorting of the lists $N_{v'}[j']$ and $N_{w'}[j']$ implies that $v$ and $w$ are either both reached through an edge corresponding to a self loop of a leaf of $CT_i$, and in this case they both have history $j+1$, or they are both reached through edges corresponding to $(v',z,\beta^{j+1})$ and $(w',z',\beta^{j+1})$, and in this case they both have history $j$, because the edges $(z,v')$ and $(z',w')$ are both labeled $2\beta^j$ by assumption.
        If $j' \neq j+1$, then $S^{-1}(v)$ and $S^{-1}(w)$ are children of $v'$ and $w'$, respectively, and hence, by \Cref{obs:ctk}, they both have history $j'+1$, and they satisfy $S^{-1}(v), S^{-1}(w) \in V_{i'} \setminus V_{i'-1}$ for some $i' \ge j' \ge i+1$.
        Hence, in both cases $v$ and $w$ have the same history. 
        
        If $j' \neq j+1$, since $i+1 > d+1$, then $v$ and $w$ have the same history, $S^{-1}(v)$ and $S^{-1}(w)$ are both in $V_{i'} \setminus V_{i'-1}$ for some $i' > d+1$, and $S^{-1}(v)$ and $S^{-1}(w)$ are connected to $CT_{i'-1}$ with edges $(v',S^{-1}(v),2\beta^{j'})$, $(S^{-1}(v),v',\beta^{j'+1})$, $(w',S^{-1}(w),2\beta^{j'})$, $(S^{-1}(w),w',\beta^{j'+1})$, for some $j' \le i'$. Hence, the second case of the invariant applies.
        
        If $j' = j+1$, then there are two possible cases: either $S^{-1}(v)$ and $S^{-1}(w)$ are both in $V_i \setminus V_{i-1}$, if we reached them by following self loops, or they both are in the clusters of $V_{i-1}$.
        
        In the former case, if $d < i-1$, then the second case applies, while if $d = i-1$, then the first case applies.
        
        Consider the latter case. The path that we used to reach $v$ (resp. $w$) must start from a node of $CT_0$, and since $CT_k$ is a tree, and since $v'$ (resp.\ $w'$) is a leaf of $CT_i$, then it means that this path already passed from $S^{-1}(v)$ (resp. $S^{-1}(w)$) previously. 
        Let $v^h$ (resp.\ $w^h$) be the node that we used to reach $v$ (resp.\ $w$) and that is at distance $d-h$ from $v_0$ (resp.\ $v_1$).
        Let $h \ge 0$ be the smallest value such that $S^{-1}(v^{h}) = S^{-1}(v)$ or $S^{-1}(w^{h}) = S^{-1}(w)$. We get that $h \ge 1$ (since $v'$ and $w'$ are in $V_i \setminus V_{i-1}$ by assumption, and hence not in the clusters of $v$ and $w$, which are in $V_{i-1}$). 
        We apply the invariant on $v^h$ and $w^h$. If the first case of the invariant holds for $v^h$ and $w^h$, then, by the definition of $h$, we get that either $S^{-1}(v) \in V_{d-h}$ or $S^{-1}(w) \in V_{d-h}$. W.l.o.g.\ assume that $S^{-1}(v) \in V_{d-h}$. We prove that $S^{-1}(w) \in V_{i-1}$. In fact, assume that it is not the case. We know that $S^{-1}{w^{h}} \in V_{d-h}$, and we also know that $S^{-1}(w^{h-1}) \not\in V_{i-1}$ by the definition of $h$. This implies that $S^{-1}(w^{h-1}) = w'$, and hence that, in order to reach $w^{h}$ from $w^{h-1}$, we need to pass from the corresponding edge of $CT_k$ that $w'$ uses to reach $S^{-1}(w)$ (they must both go towards the parent of $w'$). Hence, $S^{-1}(w^h) = S^{-1}(w) \in V_{i-1}$, which is a contradiction. 
        Hence, both $S^{-1}(v),S^{-1}(w) \in V_{i-1}$ (since $d-h \le i-1$), thus the first case of the invariant applies to $v$ and $w$.
        Finally, suppose that the second case of the invariant holds for $v^h$ and $w^h$, and hence that both $S^{-1}(v^h)$ and $S^{-1}(w^h)$ are in $V_{i'} \setminus V_{i'-1}$ for some $i' \le i-1$. Similarly as above, we get that $S^{-1}(v^h)= S^{-1}(v)$ and $S^{-1}(w^h)= S^{-1}(w)$. If $i' \le d+1$, then both $S^{-1}(v)$ and $S^{-1}(w)$ are in $V_{i'}$, and since they have the same history, then the first case of the invariant applies on $v$ and $w$. If $i' > d+1$, then the second case of the invariant applies on $v$ and $w$.
    \end{itemize}
\end{proof}

We are now ready to prove \Cref{thm:sameview}.
\begin{proof}[\bf Proof of \Cref{thm:sameview}]
    Consider two nodes mapped by the algorithm, $v$ and $w = \phi(v)$. Let $v' = S^{-1}(v)$ and $w' = S^{-1}(w)$.
    Consider the two cases of \Cref{lem:maininvariant}. In the first case, since $v',w' \in V_d$, and $d < k$, then, by the definition of cluster tree skeleton, both $v'$ and $w'$ are \emph{internal}.
    In the second case, if $i < k$, then $v$ and $w$ are also both \emph{internal}, while if $i = k$, then they are both \emph{leaves}, and they also have the same history.
    Hence, by applying \Cref{cor:sufficient} the claim follows.
\end{proof}

\begin{algorithm2e}[p!]
	\DontPrintSemicolon
	\SetKwFunction{FMain}{{\sc FindIsomorphism}}
	\SetKwProg{Fn}{Function}{:}{}
	\Fn{\FMain{$G_k$, $k$, $v_0$, $v_1$}}{
		\KwIn{A CT graph $G_k$ with $g\geq 2k+1$, $k\in \mathbb{N}$, $v_0\in C_0$, $v_1\in C_1$}
		\KwOut{Isomorphism $\phi: V(G_k^k(v_0))\rightarrow V(G_k^k(v_1))$}
		$\phi \gets$ empty map\;
		$\phi(v_0) \gets v_1$\;
		\textsc{Walk($v_0$, $v_1$, $\bot$, $k$)}\;
		\Return{$\phi$}\;
	}
	\BlankLine
	\SetKwFunction{FWalk}{{\sc Walk}}
	\SetKwProg{Pn}{Function}{:}{{\sc Walk}}
	\Pn{\FWalk{$v$, $w$, $prev$, $depth$}}{
		\If{$depth = 0$}{
			\KwRet\;
		}
		$N_v\gets$ empty list of length $k+2$\;\label{alg:iso:neighborhood-start}
		$N_w\gets$ empty list of length $k+2$\;
		\For{$i \gets 0$ \emph{\textbf{to}} $k+1$} {
			\tcp{if edge $\beta^i$ does not exist, $N_v[i]$ (resp.\ $N_w[i]$) is empty}
			$N_v[i]\gets$ list of new nodes $v'\neq prev$ found using edge $\beta^i$ from $v$, sorted such that all edges labeled $\mathsf{self}$ come before the others\; \label{code:nvi}
			$N_w[i]\gets$ list of new nodes $w'\neq \phi(prev)$ found using edge $\beta^i$ from $w$, sorted such that all edges labeled $\mathsf{self}$ come before the others\;\label{code:nwi}\label{alg:iso:neighborhood-end}
		}
		\textsc{Map($N_v$, $N_w$)}\;\label{alg:iso:call-map}
		\For{$i \gets 0$ \emph{\textbf{to}} $k+1$} {
			\For{$v'$ \emph{\textbf{in}} $N_v[i]$}{
				\textsc{Walk($v'$, $\phi(v')$, $v$, $depth-1$)}
			}
		}
	}
	\BlankLine
	\SetKwFunction{FPair}{{\sc Map}}
	\SetKwProg{FPn}{Function}{:}{{\sc Map}}
	\FPn{\FPair{$N_v$, $N_w$}}{
		\For{$i \gets 0$ \emph{\textbf{to}} $k+1$\label{alg:iso:map-outer-for-start}} {
			\tcp{$zip(\cdot,\cdot)$ yields element tuples until the shorter list ends}
			\For{$v',w'$ \emph{\textbf{in}} $zip(N_v[i], N_w[i])$}{
				$\phi(v')\gets w'$\;\label{alg:iso:map-inner-for-end}
			}
		}
	\tcp{$len(\cdot)$ returns the length of a list}
	\If{$\exists~i\in [k+1]_0:  len(N_v[i]) \neq len(N_w[i])$\label{alg:iso:special-start}}{
	\tcp{we will prove that $len(L_v[i]) = len(L_w[i])$ for $i\in[k+1]_0\setminus \{i_v,i_w\}$}
	$i_v \gets i \in [k+1]_0: len(N_v[i]) = len(N_w[i]) + 1$\;
	$i_w \gets i \in [k+1]_0: len(N_v[i]) + 1 = len(N_w[i])$\;
	\tcp{$L[i][-1]$ retrieves the last element from list $i$ in~$L$}
	$\phi(N_v[i_v][-1]) \gets N_w[i_w][-1]$\;\label{alg:iso:special-end}
}
	}
	\caption{This is the algorithm presented in \cite{breezing}, that finds an isomorphism $\phi: V(G_k^k(v_0))\rightarrow V(G_k^k(v_1))$. The only lines that differ are \ref{code:nvi} and \ref{code:nwi}.}\label{algo:isomorphism}
\end{algorithm2e}

\subsection{Lifting}
\begin{proof}[\bf Proof of \Cref{lemma:randomlift}]
    We first prove that the probability for a node $\tilde{v}$ to be contained in a cycle of length at most $\ell\geq 3$ is upper bounded by $\Delta^\ell/q$. Let us first look at the probability that $\tilde{v}$ is contained in some cycle of length exactly $\ell'\in\set{3,\dots,\ell}$. For $\tilde{v}$ to be contained in a cycle of length $\ell'$ in $\tilde{G}$, one of the cycles of length $\ell'$ of $v$ in $G$ (that could pass from the same nodes or edges multiple times) has to be lifted to a cycle of length $\ell'$ of $\tilde{v}$ in $\tilde{G}$. Consider such a cycle $v=v_1, v_2, \dots, v_{\ell'}$ in $G$. In $\tilde{G}$, $\tilde{v}=\tilde{v}_1$ is connected to exactly one node $\tilde{v}_2$ of the fiber of $v_2$. Further, $\tilde{v}_2$ is connected to exactly one node $\tilde{v}_3$ in the fiber of $v_3$, and so on. There is therefore exactly one path of the form $\tilde{v}_1,\tilde{v}_2,\dots,\tilde{v}_{\ell'}$, where $\tilde{v}_1=\tilde{v}$ and where every node $\tilde{v}_i$ is in the fiber of node $v_i$ of $G$. For this path to extend to a cycle in $\tilde{G}$, the edge $\set{\tilde{v}_{\ell'},\tilde{v}_1}$ needs to be present in $\tilde{G}$. Because the edges in $\tilde{G}$ between the nodes in the fibers of $v_{\ell'}$ and $v_1$ are formed by a uniform random perfect matching, the probability for the edge $\set{\tilde{v}_{\ell'},\tilde{v}_1}$ to be present in $\tilde{G}$ is exactly $1/q$. For every cycle of length $\ell'$ of $v$ in $G$, the probability for the cycle to appear as a cycle of $\tilde{v}$ in $\tilde{G}$ is therefore exactly $1/q$. The number of cycles of length exactly $\ell'$ of nodes $v$ in $G$ can be upper bounded by $\Delta^{\ell'-1}$ and by a union bound, the probability for $\tilde{v}$ to be contained in a cycle of length $\ell'$ in $\mathcal{G}$ can consequently be upper bounded by $\Delta^{\ell'-1}/q$. The upper bound of $\Delta^\ell/q$ on the probability for $\tilde{v}$ to be contained in some cycle of length at most $\ell$ in $\mathcal{G}$ now follows by a union bound over all $\ell'\in\set{3,\dots,\ell}$.
    
    Consider a set $\tilde{S}\subseteq\tilde{C}$ of size $|\tilde{S}|=s\cdot q$, where $s$ is an integer with $s\geq 8\ln|C|$. We want to upper bound the probability that the nodes in $\tilde{S}$ form an independent set of $\tilde{G}$. For this, we partition $\tilde{S}$ according to the different fibers that are contained in $\tilde{S}$. Assume that $\tilde{S}=\tilde{S}_1,\dots,\tilde{S}_t$, where for every $i\in \set{1,\dots,t}$, all nodes in $\tilde{S}_i$ are from the same fiber and for any $i\neq j$, any two nodes in $\tilde{S}_i$ and $\tilde{S}_j$ are from different fibers. Assume further that for all $i\in \set{1,\dots,t}$, $|\tilde{S}_i|=\sigma_i$. Note that there clearly are no edges between two nodes of the same set $\tilde{S}_i$ (for any $i$). Further, the edges in $\tilde{G}$ between the different pairs of sets $\tilde{S}_i$ and $\tilde{S}_j$ are chosen independently (because different edges of $G$ lead to independent random perfect matchings between the corresponding fibers in $\tilde{G}$). 
    
    Let us therefore first concentrate on a single pair of sets $\tilde{S}_i$ and $\tilde{S}_j$ for $i\neq j$. Assume that $\tilde{S}_i$ is part of the fiber of some node $v_i\in V$ of $G$ and that $\tilde{S}_j$ is part of the fiber of some other node $v_j\in V$ of $G$. We enumerate the nodes in $\tilde{S}_i$ as $\tilde{S}_i=\set{\tilde{v}_{i,1},\dots,\tilde{v}_{i,\sigma_i}}$. The neighbor of $\tilde{v}_{i,1}$ in the fiber of $v_j$ is a uniformly random node of the $q$ nodes in the fiber. The probability that the neighbor of $\tilde{v}_{i,1}$ is outside $\tilde{S}_j$ is therefore equal to $(q-\sigma_j)/q=1-\sigma_j/q$. Conditioning on the event that $\tilde{v}_{i,1}$ is connected to a node in the fiber of $v_j$ outside $\tilde{S}_j$, the probability that the neighbor of $\tilde{v}_{i,2}$ is outside $\tilde{S}_j$ is equal to $(q-1-\sigma_j)/(q-1)=1-\sigma_j/(q-1)$ (or equal to $0$ if $\sigma_j\geq q-1$). More generally, the probablity that $\tilde{v}_{i,h}$ (for $h\in \set{1,\dots,\sigma_i}$) is connected to a node outside $\tilde{S}_j$ in the fiber of $v_j$ given that the nodes $\tilde{v}_{i,1},\dots,\tilde{v}_{i,h-1}$ are connected to nodes outside $\tilde{S}_j$ is equal to $1-\min\set{1,\sigma_j/(q-h+1)}$. Let $P_{i,j}$ be the probability that there is no edge between $\tilde{S}_i$ and $\tilde{S}_j$ in $\tilde{G}$. We have
    \[
    P_{i,j} = \prod_{h=1}^{\sigma_i} 1 - \min\set{1,\frac{\sigma_j}{q-h+1}} \leq \left(1-\frac{\sigma_j}{q}\right)^{\sigma_i}
    < e^{-\sigma_i\sigma_j/q}.
    \]
    The probability that $\tilde{S}$ is an independent set of $\tilde{G}$ can therefore be upper bounded by
    \[
    \Pr(\tilde{S}\text{ is indep.\ set of }\tilde{G}) = \prod_{1\leq i<j\leq t} P_{i,j} <
    \exp\left(-\frac{1}{q}\cdot\sum_{1\leq i<j\leq t} \sigma_i\cdot\sigma_j\right). 
    \]
    The sum over the $\sigma_i\sigma_j$ for all $1\leq i<j\leq t$ is equal to the total number of unordered pairs of nodes in $\tilde{S}$ such that the two nodes are in different fibers. Because a fiber only consists of $q$ nodes, for every node $\tilde{u}\in\tilde{S}$, there are at least $|\tilde{S}|-q=(s-1)q$ nodes $\tilde{v}\in\tilde{S}$ such that $\tilde{u}$ and $\tilde{v}$ are in different fibers. The total number of unordered pairs of nodes from different fibers in $\tilde{S}$ is therefore at least $sq\cdot(s-1)q/2$. Since we assume that $|C|\geq 2$ and $s\geq 8\ln|C|$, this number is lower bounded by $s^2q^2/4$. We thus get
    \[
    \Pr(\tilde{S}\text{ is indep.\ set of }\tilde{G}) < e^{-s^2q/4}.
    \]
    By a union bound over all possible sets $\tilde{S}$, the probability that $\tilde{G}[\tilde{C}]$ contains an independent set of size $s\cdot q$ can then be upper bounded by
    \[
    \binom{|\tilde{C}|}{sq}\cdot e^{-s^2q/4} = \binom{|C|q}{sq}\cdot e^{-s^2q/4} \leq 
    \left(\frac{e|C|}{s}\right)^{sq}\cdot e^{-s^2q/4} = e^{sq(1+\ln|C| - \ln s - s/4)}
    < e^{-s^2q/8}.
    \]
    The last inequality follows from $s\geq 8\ln|C|$. This implies that $\ln|C|-s/4\leq -s/8$ and it also implies that $1-\ln s < 0$.
\end{proof} 

\subsection{(Almost) High-Girth Graphs in \boldmath$\mathcal{G}_k$}
\label{app:missingHighGirth}

\begin{proof}[\bf Proof of \Cref{lem:smallgirth}]
    The bound on the independence number comes from the fact that each cluster is composed of disjoint cliques of size $\beta^i$, plus additional edges, and from the fact that, for each clique, at most one node can join the independent set.
    
    The bound on the maximum degree comes from the fact that internal nodes have exactly $\sum_{i=0}^{k}2\beta^{i}\le 2\beta^{k+1}$ neighbors, while leaves have exactly $2\beta^{k+1}$ neighbors (by \Cref{obs:ctk}).
    
    Let $T_i$ be the set of nodes of $CT_k$ at distance exactly $i$ from $v_0$. Note that, by construction, $|T_{i+1}| \le (k+1) |T_{i}|$. Consider now clusters of nodes of $G_k$ corresponding to nodes of $T_i$: they have size $2 \beta^{k+1} (\beta/2)^{k+1-i}$. Hence, the clusters at distance $i+1$ are smaller than the ones at distance $i$ by a factor $\beta/2$. 
    Let $S_i = \bigcup_{v : \mathrm{dist}(v,v_0)= i} S(v)$.
    We obtain the following:
    \[
        \frac{|S_{i+1|}}{|S_i|} \le 2(k+1) / \beta,
    \]
    that is, the ratio between the number of nodes of $G_k$ corresponding to nodes at distance $i+1$ from $v_0$, and the ones at distance $i$, is at most $2(k+1) / \beta$. Hence, the total number of nodes of $G_k$ is at most:
    \[
        2 \beta^{k+1} (\beta/2)^{k+1} \sum_{i=0}^{k+1} \left(\frac{2(k+1)}{\beta}\right)^i,
    \]
    and since $2(k+1)/\beta < 1/2$, then the number of nodes is $O(\beta^{2k+2})$.
\end{proof}

\begin{proof}[\bf Proof of \Cref{lem:lifted}]
    We apply \Cref{lemma:randomlift} on the the graph $G_k$ described in \Cref{lem:smallgirth}. Recall that each cluster $S(v)$ of $G_k$ corresponding to some node $v\in N(v_0)$ satisfying that $i=\psi(v)$, contains $t$ disjoint cliques of size $\beta^i$. Consider an arbitrary such clique $C$.
    \Cref{lemma:randomlift} guarantees that the graph induced by the set $\tilde{C}$ of nodes corresponding to copies of nodes of $C$ has independence number $\alpha(\tilde{G}_k[\tilde{C}]) \le s q$ with probability at least $1-e^{-s^2 q /8}$ for every integer $s \ge 8 \ln \beta^i$.
    Hence, the probability that all cliques of $S(v)$ have independence number at most $sq$ is at least $(1-e^{-s^2 q /8})^{t}$. Since $S(v)$ is a union of these cliques, plus some additional edges, then this value is also a lower bound on the probability that the independence number of $S(v)$ is at most $s q t$, as desired.
\end{proof}

\begin{proof}[\bf Proof of \Cref{cor:family}]
We apply \Cref{lem:lifted} with $q=\beta^{ck^2}$, for some large enough constant $c \ge 1$.
Let $v \in N(v_0)$ and $i = \psi(v)$. By fixing $s=\lceil 8\ln \beta^i \rceil$ we obtain that the probability that a cluster satisfies the independence requirements is at least
 \[
    (1-e^{-s^2q/8})^t \ge 1-te^{-s^2q/8} \ge 1 - \beta^{2k-i+1 -s^2 q /8} \ge 1 - \beta^{2k-8q \ln^2 \beta },
 \]
 where the last inequality holds because $i \ge 1$.
 
 Note that for $q \ge k$, $\beta^{2k-8q \ln^2 \beta} < 1$.
The probability that all clusters that are neighbors of $S(v_0)$ satisfy the requirements is then at least
\[
    (1 - \beta^{2k-8q \ln^2 \beta })^{k+1} \ge 1 - (k+1)\beta^{2k-8q \ln^2 \beta} \ge 1 - \beta^{2k+\log_\beta (k+1)-8q \ln^2 \beta} > 0,
\]
where the last inequality holds by the assumption on the value of $q$. Hence, there is non-zero probability to obtain a graph satisfying the requirements, and hence such a graph exists. We get that this graph satisfies that each cluster $S(v)$ has independence number at most 
\[
 s q t = 8 \ln \beta^i q 2 \beta^{k-i+1} (\beta/2)^{k} =  |S(v)| 8 \ln \beta^i / \beta^i,
\]
as required.

For a node $v$, the probability that it is contained in a cycle of length at most $2k+1$ is at most 
\[
    \frac{\Delta^{2k+1}}{q} = \frac{(2\beta^{k+1})^{2k+1}}{\beta^{ck^2}} = 2^{2k+1} \beta^{(k+1)(2k+1) - ck^2} \le 1/\beta,
\]
for a large enough constant $c$.

The only thing remaining to be done is to determine the total number of nodes, which is $O(q \beta^{2k+2}) = O(\beta^{ck^2 + 2k+2} ) = \beta^{O(k^2)}$, as desired.
\end{proof}

\subsection{Lower Bound on Node-Averaged Complexity of Maximal Matching}
\label{app:lb_avgMatching}

\begin{proof}[\bf Proof \Cref{thm:lb_avgMatching}]
    To prove the lower bound on the node-averaged complexity of the maximal matching problem, we use the variant of the KMW graph construction that has been used in \cite{kuhn16_jacm} for proving a worst-case lower bound on approximating the maximum matching problem. The graph consists of $2$ copies of the cluster graph construction as described in \Cref{sec:clustertrees}. Both copies contain independent sets $C_0$ and $C_1$ such that every node in $C_0$ has one neighbor in $C_1$ and every node in $C_1$ has $\delta=\omega(1)$ neighbors in $C_0$. In the following, we use $C_0$ and $C_0'$ to refer to the sets in the two copies and we similarly use $C_1$ and $C_1'$. In addition, the graph contains a perfect matching between the two copies of the basic cluster graph such that each edge of this perfect matching connects a node in the first copy with a node in the same cluster of the second copy. The graph has the property that every matching that does not contain any of the edges of the perfect matching between the two copies is of size $o(n)$. Further, the two sets $C_0$ and $C_0'$ together contain $(1-o(1))n$ nodes. Finally, the graph is constructed with no cycles of length less than $2k+1$ and it is shown in \cite{kuhn16_jacm} that in this case, all edges between clusters $C_0$, $C_1$, $C_0'$, and $C_1'$ have the same $k$-hop views. Consequently, after $k$ rounds, all those edges need to be added to the matching with the same probability. Because nodes in $C_1$ and $C_1'$ are incident to $\delta=\omega(1)$ such edges, this probability has to be $o(1)$. Because $(1-o(1))n/2$ of the perfect matching edges are between $C_0$ and $C_0'$, this implies that any matching after $k$ rounds has to be of size $o(n)$. However at the end, any maximal matching must contain $(1-o(1))n/2$ perfect matching edges between $C_0$ and $C_0'$. This implies that after $k$ rounds, most of the nodes in $C_0$ and $C_0'$ (and thus $(1-o(1))n$ nodes) are all not decided and thus, the node-averaged complexity is at least $k=\Omega\Big(\min\Big\{\frac{\log\Delta}{\log\log\Delta},\sqrt{\frac{\log n}{\log\log n}}\Big\}\Big)$.
\end{proof}

\end{document}